\documentclass[preprint,3p,times]{elsarticle}
\RequirePackage{multirow,booktabs,subfigure,color,array,hhline,makecell}
\usepackage{amssymb}
\usepackage{amsmath}
\usepackage{graphicx}
\usepackage{amsthm}
\usepackage{mathrsfs}
\usepackage{indentfirst}
\usepackage[colorlinks,citecolor=blue,urlcolor=blue]{hyperref}
\allowdisplaybreaks
\usepackage[table]{xcolor}
\usepackage{tikz-network}
\usepackage{pgf}
\usepackage{tikz}
\usepackage{subfigure}
\usepackage[T1]{fontenc}\usepackage[utf8]{inputenc}
\usetikzlibrary{arrows, decorations.pathmorphing, backgrounds, positioning, fit, petri, automata}
\usetikzlibrary{shadows,arrows,positioning}
\RequirePackage{algorithm,algpseudocode}
\allowdisplaybreaks
\usepackage{changes}

\DeclareMathOperator{\rank}{rank}   
\usepackage{setspace}
\newtheorem{thm}{Theorem}
\newtheorem{defin}{Definition}
\newtheorem{lem}{Lemma}

\newtheorem{rem}{Remark}
\newtheorem{cor}{Corollary}

\allowdisplaybreaks[4]
\DeclareMathOperator{\diag}{diag} 
\AtBeginDocument{%
	\providecommand\BibTeX{{%
			\normalfont B\kern-0.5em{\scshape i\kern-0.25em b}\kern-0.8em\TeX}}}
\journal{~}

\begin{document}

\begin{frontmatter}
\title{Exact Community Recovery in Bipartite Networks}
\author[label1]{Huan Qing\corref{cor1}}
\ead{qinghuan@u.nus.edu~\&~qinghuan@cqut.edu.cn}
\cortext[cor1]{Corresponding author.}
\address[label1]{School of Economics and Finance, Chongqing University of Technology, Chongqing, 400054, China}

\begin{abstract}
Community detection in bipartite networks is a fundamental problem in modern data analysis, with applications in recommendation systems, biological networks, and social network analysis. Unlike conventional unipartite graphs, bipartite networks consist of two distinct types of nodes with edges only connecting across types, so recovering latent communities requires estimating labels on the two node types. The stochastic co-blockmodel is a classical probabilistic framework for such networks, yet theoretical guarantees for exact community recovery in this setting remain limited, especially when the number of communities grows, the community sizes are unbalanced, or the degrees are heterogeneous. In this work, we prove that a simple spectral clustering algorithm based on the diagonal-deleted Gram matrix achieves exact recovery with high probability under mild conditions on sparsity, community balance, and the number of clusters. We further extend the result to the degree-corrected stochastic co-blockmodel, where each node carries its own degree heterogeneity parameter, and show that a row-normalized version of the same algorithm maintains the exact recovery guarantee. Extensive experiments validate our theoretical findings.
\end{abstract}

\end{frontmatter}

\section{Introduction}\label{sec:intro}
Bipartite networks, also known as two-mode networks, consist of two distinct types of nodes, with edges only connecting nodes of different types \citep{everett2013dual, latapy2008basic}. Such networks arise in a wide range of applications, including voter-candidate relationships in election campaigns \citep{doreian2013partitioning, kleinnijenhuis2013adjustment}, customer-product and customer-merchant interactions in online platforms and recommendation systems \citep{huang2007analyzing, martens2016mining}, fund-stock and investor-asset networks in finance \citep{souma2003complex, marotta2015bank,squartini2017enhanced}, author-paper relationships in citation and scientific collaboration networks \citep{newman2001structure, ji2016coauthorship}, and plant-pollinator networks in ecology \citep{young2021reconstruction}. Community detection is a fundamental task in network analysis \citep{fortunato2010community,MALLIAROS201395,fortunato2016community}. In bipartite settings, the task seeks to partition nodes on both sides of the bipartition into latent groups so that nodes within the same group share similar connection patterns to the other side. Uncovering these groups helps reveal the underlying functional structure of complex systems and supports further modeling and inference.

To capture the latent community structures in bipartite networks, \cite{rohe2016co} introduced the stochastic co-blockmodel (ScBM) as a generalization of the classical stochastic blockmodel (SBM) introduced by \cite{holland1983stochastic} from undirected networks to bipartite networks. In ScBM, each node is assigned to a latent community, and the probability of an edge from a node on the sending side (rows) to a node on the receiving side (columns) depends solely on their respective community memberships. Recognizing that real-world networks exhibit degree heterogeneity, \cite{rohe2016co} also proposed the degree-corrected stochastic co-blockmodel (DCScBM), which extends the popular degree-corrected stochastic blockmodel (DCSBM) developed by \cite{karrer2011stochastic} from undirected networks to bipartite settings. In DCScBM, each node is assigned an individual scale parameter, allowing for highly heterogeneous degree distributions within communities.

A range of algorithms have been developed for these models. Spectral methods are popular due to their computational efficiency, strong empirical performance, and tractability for theoretical analysis \citep{von2007tutorial}. In the single-part setting, spectral methods have been extensively studied for SBM and DCSBM, where consistency and strong consistency have been established under various sparsity and degree-heterogeneity regimes in \cite{rohe2011spectral,lei2015consistency,qin2013regularized,jin2015fast,joseph2016impact,su2020strong,su2021jmlr,jing2022community,deng2024strong}; exact recovery guarantees have also been obtained for SBM in \cite{lei2020unified,abbeentry2020}. For a comprehensive overview of community detection in SBM, see \cite{abbe2018community}; for a statistical perspective on spectral methods, see \cite{chen2021spectral}. For ScBM, spectral clustering algorithms have been shown to achieve estimation consistency by \cite{zhou2019analysis}. For directed networks, which can be viewed as a special case of bipartite networks, \cite{wang2020spectral} proposed the D-SCORE algorithm, building upon the SCORE framework of \cite{jin2015fast}, and proved its consistency under the directed DCSBM. For DCScBM, algorithms developed in \cite{rohe2016co,qing2023community} have been shown to achieve estimation consistency. \cite{guo2023randomized} developed randomized spectral co-clustering algorithms for large-scale directed networks, establishing approximation error bounds under ScBM and DCScBM while reducing computational cost via random projection and random sampling techniques. 

However, it remains unclear whether any algorithm can achieve exact recovery---meaning zero misclassification error with high probability---in ScBM and its degree-corrected extension. Existing exact recovery results have been developed for various model variants under different assumptions. \cite{zhou2020optimal} characterized the minimax optimal weak consistency under ScBM, with exact recovery only arising in a logarithmic sparsity regime. \cite{ndaoud2022TIT} established exact recovery for one side of the bipartition in the two-community ScBM using a spectral method with iterative refinement, while accommodating moderate community imbalance. \cite{braun2023minimax} developed a generalized iterative algorithm with spectral initialization under ScBM, allowing for arbitrary numbers of communities on each side, and established exact recovery guarantees. \cite{braun2024strong} proved that a simple spectral algorithm suffices for exact recovery of the row communities in ScBM under the optimal sparsity regime, provided that the communities are approximately balanced. Despite these advances, none of these results cover the degree-corrected stochastic co-blockmodel, which accommodates individual node heterogeneity. Moreover, they do not provide exact recovery guarantees under mild conditions on number of communities, sparsity, and community balance. These observations motivate our development of spectral methods that achieve exact recovery under ScBM and DCScBM.

In this paper, we develop efficient spectral algorithms for exact community recovery in ScBM and its degree-corrected version. Under ScBM, our approach constructs the diagonal-deleted Gram matrix, extracts its leading eigenvectors, and applies k-means to recover the communities. We show that the proposed method exactly recovers the community labels in the general ScBM setting, where our conditions allow the number of communities to grow with the network size and characterize the trade-off between sparsity, community balance, and other model parameters. We then extend the analysis to DCScBM, where we employ a row-normalized variant of the same spectral procedure to accommodate degree heterogeneity, and show that the exact recovery guarantee is retained under mild conditions on the degree parameters. When the degree-correction parameters are all equal, our DCScBM results reduce to the ScBM guarantees. Numerical experiments are conducted to validate our theoretical findings.

The remainder of the paper is organized as follows. Section~\ref{sec:SCBMrecovery} formalizes ScBM, presents the spectral clustering algorithm, and states the main exact recovery theorem. Section~\ref{sec:DCScBM} extends the results to DCScBM. Section~\ref{sec:num} presents numerical experiments, and Section~\ref{sec:conclusion} concludes the paper. All technical proofs are collected in the appendix.

\paragraph*{Notation}
For a vector \(v\), \(\|v\|_2\) denotes its Euclidean norm. For any matrix \(M\), \(\|M\|\) denotes its spectral norm (largest singular value), \(\|M\|_F\) denotes its Frobenius norm, \(\|M\|_\infty\) denotes its entrywise maximum absolute value, \(\|M\|_{2,\infty} := \max_i \|M_{i,:}\|_2\) denotes its maximum row Euclidean norm, \(M^\top\) denotes its transpose, \(\sigma_k(M)\) denotes its \(k\)-th largest singular value, \(\operatorname{rank}(M)\) denotes its rank, and \(\det(M)\) denotes its determinant. \(\operatorname{diag}(a_1,a_2,\dots,a_p)\) is the diagonal matrix with entries \(a_i\), and \(I_p\) is the \(p\times p\) identity matrix. The vector \(e_k\) denotes the \(k\)-th standard basis vector. \(\mathbb{R}\) denotes the set of real numbers. For a random variable \(X\), \(\mathbb{E}[X]\) denotes its expectation. We write \(a_N \gg b_N\) if \(b_N/a_N \to 0\), \(a_N \succeq b_N\) if \(a_N \ge C b_N\) for some absolute constant \(C>0\), \(a_N \preceq b_N\) if \(a_N \le C b_N\) for some absolute constant \(C>0\), and \(a_N \asymp b_N\) if \(a_N \succeq b_N\) and \(b_N \succeq a_N\).
\section{Stochastic Co-Blockmodels}\label{sec:SCBMrecovery}

This section develops the theoretical guarantees for exact community recovery in bipartite networks. We first introduce the stochastic co-blockmodel (ScBM), which accommodates asymmetric interactions between two distinct node types. We then establish structural properties of the expected adjacency matrix that justify the use of spectral methods. A simple spectral clustering algorithm is presented, followed by our main result: explicit conditions under which the algorithm recovers the true communities exactly with high probability. 
\subsection{Model and Algorithm}

We consider a bipartite network with row nodes (senders) and column nodes (receivers). Edges exist only from row nodes to column nodes, and there are no edges within the same node type. We now formally introduce the stochastic co-blockmodel (ScBM) of \cite{rohe2016co} in the following definition.
\begin{defin}\label{def:ScBM}
A \emph{stochastic co-blockmodel} (ScBM) with \(n_y\) row nodes, \(n_z\) column nodes, \(K_y\) row communities, and \(K_z\) column communities, where \(1\le K_y\le K_z\), is parameterized by:
\begin{itemize}
\item membership matrices \(Y\in\{0,1\}^{n_y\times K_y}\) and \(Z\in\{0,1\}^{n_z\times K_z}\), where \(y_i\in[K_y]\) and \(z_j\in[K_z]\) denote the community memberships of the \(i\)-th row node and the \(j\)-th column node, respectively; each row of \(Y\) and \(Z\) contains exactly one 1, so that \(Y_{ik}=\mathbf{1}\{y_i=k\}\) and \(Z_{jl}=\mathbf{1}\{z_j=l\}\);
\item a connectivity matrix \(B\in[0,1]^{K_y\times K_z}\) satisfying \(\rank(B)=K_y\) and \(\sigma_B:=\sigma_{K_y}(B)>0\), where \(\sigma_B\) controls the separation between community connectivity profiles;
\item a sparsity parameter \(\rho\in(0,1]\).
\end{itemize}
The observed adjacency matrix \(A\in\{0,1\}^{n_y\times n_z}\) has independent entries
\begin{align*}
A_{ij}\sim \mathrm{Bernoulli}\bigl(\rho B_{y_i,z_j}\bigr),\quad 1\le i\le n_y,\;1\le j\le n_z. 
\end{align*}
The expected adjacency matrix is \(\Omega:=\mathbb{E}[A]=\rho YBZ^\top\), and \(\rank(\Omega)=K_y\).
\end{defin}

Throughout this paper, we assume that each community contains at least one node. The condition \(\rank(B)=K_y\) ensures that the \(K_y\) row communities have distinct connectivity profiles, which is necessary for the recovery of the row clusters. When \(K_y<K_z\), the column communities are still identifiable in the sense that the model parameters (including \(Z\)) are uniquely determined up to label permutations as shown in \cite{rohe2016co}. However, exact recovery of all column community labels is impossible because the row space of \(\Omega\) has dimension \(K_y<K_z\), so the observed signal only provides a \(K_y\)-dimensional projection of the column structure. We return to this point in Remark~\ref{rem:identifiability} below.

We adopt the following notations throughout. For each row community \(k\in[K_y]\), let \(n_k^y\) denote the number of row nodes in row community \(k\); similarly, for each column community \(l\in[K_z]\), let \(n_l^z\) be the number of column nodes in column community \(l\). Define
\begin{align*}
&n_{\min}^y:=\min_k n_k^y,\quad n_{\max}^y:=\max_k n_k^y,\quad
n_{\min}^z:=\min_l n_l^z,\quad n_{\max}^z:=\max_l n_l^z,\quad N:=n_y+n_z,\\
&\beta_{y}:=\frac{K_y n_{\min}^y}{n_y}, \quad 
\beta_{z}:=\frac{K_z n_{\min}^z}{n_z},\quad
\tau_y := \frac{n_{\max}^y}{n_{\min}^y},\quad
\tau_z := \frac{n_{\max}^z}{n_{\min}^z},\quad 
\kappa_B := \frac{\sigma_1(B)}{\sigma_{K_y}(B)}.
\end{align*}
The quantities \(\beta_y\) and \(\beta_z\) measure how balanced the community sizes are on each side of the bipartite graph. If all communities on a given side have equal size, the corresponding \(\beta\) equals \(1\); smaller values indicate greater imbalance. The ratios \(\tau_y\) and \(\tau_z\) capture the disparity between the largest and smallest communities on each side, with \(\tau_y=\tau_z=1\) corresponding to perfect balance. The parameter \(\kappa_B\) is the condition number of the connectivity matrix \(B\). All logarithmic factors in this paper are taken with respect to \(N\); when we specialize to the balanced regime, we write \(n_y\asymp n_z\asymp n\), so that \(N\asymp n\). 

Let \(\Omega=U\Sigma V^\top\) be the compact singular value decomposition, with \(U\in\mathbb{R}^{n_y\times K_y}\), \(V\in\mathbb{R}^{n_z\times K_y}\), and \(\Sigma=\diag(\sigma_1,\dots,\sigma_{K_y})\) where \(\sigma_1\ge\cdots\ge\sigma_{K_y}>0\), $U^\top U=I_{K_y}$, and $V^\top V=I_{K_y}$. Define
\begin{align*}
\sigma_{\min}:=\sigma_{K_y}(\Omega),\quad \sigma_1:=\sigma_1(\Omega),\quad \kappa_{\Omega}:=\frac{\sigma_1}{\sigma_{\min}}. 
\end{align*}

The spectral approach relies on the key observation that the singular vectors of \(\Omega\) are constant on communities. The following lemma formalizes this property and provides explicit distances between community centroids. These distances will be used later to guarantee that the rows of the estimated singular vectors are sufficiently  separated, enabling exact recovery via \(k\)-means.

\begin{lem}\label{lem:SVstructure}
Under Definition~\ref{def:ScBM}, let \(\Omega=U\Sigma V^\top\) be the compact SVD. Then there exists an invertible matrix \(X\in\mathbb{R}^{K_y\times K_y}\) and a full column rank matrix \(W\in\mathbb{R}^{K_z\times K_y}\) such that
\begin{align*}
U=YX,\quad V=ZW. 
\end{align*}
Moreover, for any distinct row communities \(k\ne l\), we have
\begin{align*}
\|X_{k,:}-X_{l,:}\|_2 = \sqrt{\frac{1}{n_k^y}+\frac{1}{n_l^y}}. 
\end{align*}
If \(K_y=K_z\), then \(W\) is square and invertible, and for any distinct column communities \(k\ne l\), we have
\begin{align*}
\|W_{k,:}-W_{l,:}\|_2 = \sqrt{\frac{1}{n_k^z}+\frac{1}{n_l^z}}. 
\end{align*}
\end{lem}

The distance formula \(\sqrt{\frac{1}{n_k^y}+\frac{1}{n_l^y}}\) implies that the minimum separation between distinct row community centroids is at least \(\sqrt{2/n_{\max}^y}\). This lower bound determines the margin required for the nearest-centroid classifier to succeed.

We now describe the core spectral algorithm. To eliminate the bias caused by noise, we remove the diagonal entries of the sample Gram matrix. Define
\[
G_y := \mathcal{P}_{\mathrm{off-diag}}(AA^\top),
\]
where $\mathcal{P}_{\mathrm{off-diag}}(\cdot)$ zeroes out all diagonal entries. The deletion ensures that the expectation of $G_y$ is free of noise variance, and its population version differs from $\Omega\Omega^\top$ only by the diagonal bias $\mathcal{P}_{\mathrm{diag}}(\Omega\Omega^\top)$. Under the conditions of our main theorem, Lemma~1 of \cite{cai2021subspace} guarantees that this bias is spectrally negligible relative to the signal strength, so the leading eigenspace of $G_y$ is provably close to the column space of $\Omega$ spanned by its left singular vectors $U$. This justifies using the leading eigenvectors of $G_y$ for subspace estimation. We term this procedure \emph{Spectral Co-clustering with Diagonal Deletion} (SCDD), summarized in Algorithm~\ref{alg:SCDD}. The algorithm primarily recovers the row communities by extracting the leading $K_y$ eigenvectors of $G_y$ and applying $k$-means to their rows. When $K_y=K_z$, the column communities are recovered symmetrically by applying the identical procedure to $G_z := \mathcal{P}_{\mathrm{off-diag}}(A^\top A)$.
\begin{algorithm}[!ht]
\caption{Spectral Co-clustering with Diagonal Deletion (SCDD) for bipartite networks}
\label{alg:SCDD}
\begin{algorithmic}[1]
\Require Adjacency matrix \(A\in\{0,1\}^{n_y\times n_z}\), number of row communities \(K_y\) (and optionally \(K_z\) if column recovery is desired).
\Ensure Row labels \(\hat y\in[K_y]^{n_y}\).
\State Compute the diagonal-deleted Gram matrix \(G_y:=\mathcal{P}_{\mathrm{off-diag}}(AA^\top)\in\mathbb{R}^{n_y\times n_y}\).
\State Let \(\widehat U\widehat\Lambda\widehat U^\top\) be the top-\(K_y\) eigen-decomposition of \(G_y\), where \(\widehat U\in\mathbb{R}^{n_y\times K_y}\) has orthonormal columns.
\State Run the \(k\)-means algorithm on the rows of \(\widehat U\) with \(K_y\) clusters and return the resulting estimated labels \(\hat y\).
\State \textbf{Optional (column recovery):} If \(K_y=K_z\), compute \(G_z:=\mathcal{P}_{\mathrm{off-diag}}(A^\top A)\), take its leading \(K_y\) eigenvectors \(\widehat V\), run \(k\)-means on the rows of \(\widehat V\) with \(K_y\) clusters, and return \(\hat z\).
\end{algorithmic}
\end{algorithm}

\subsection{Exact Recovery Guarantee}
Our goal is to give a set of conditions under which the SCDD algorithm recovers the true row communities exactly, with probability tending to one as the network grows. These conditions are stated directly in terms of the model parameters and a few incoherence quantities associated with the expected adjacency matrix \(\Omega\). Define

\[
\mu_0:=\frac{n_y n_z\|\Omega\|_\infty^2}{\|\Omega\|_F^2},\quad
\mu_1:=\frac{n_y}{K_y}\|U\|_{2,\infty}^2,\quad
\mu_2:=\frac{n_z}{K_y}\|V\|_{2,\infty}^2,\quad
\mu:=\max\{\mu_0,\mu_1,\mu_2\},
\]
where $\mu_0$ controls the maximum entrywise magnitude of $\Omega$, while $\mu_1$ and $\mu_2$ characterize how evenly the energy of its left and right singular vectors is distributed across rows and columns, respectively. Such incoherence measures are standard in the subspace estimation literature, as they guarantee that the signal is diffuse---avoiding concentrated entries or singular vectors.

To formalize the notion of successful recovery, we define a clustering loss that counts the minimum number of row nodes whose estimated community labels differ from the true ones, after optimally matching the community labels. Let

\[
\ell(\hat y,y):=\min_{\pi\in\mathcal{S}_{K_y}}\sum_{i=1}^{n_y}\mathbf{1}\{\hat y_i\neq \pi(y_i)\},
\]
where \(\mathcal{S}_{K_y}\) denotes the set of all permutations of the \(K_y\) community labels. Thus, \(\ell(\hat y,y)=0\) means that the estimated labels agree with the true ones up to a global relabeling, which is exactly the notion of exact community recovery.

The following theorem provides explicit conditions, expressed through the model parameters and the incoherence measures above, under which the SCDD algorithm recovers the true row communities exactly with high probability.

\begin{thm}\label{thm:mainFinal}
Under ScBM, suppose the following conditions hold:
\begin{align}
&n_y n_z \succeq \mu^2 \kappa_B^8 \tau_y^4 \tau_z^4 K_y^2 \log^4 N, \label{cond:1}\\
&n_z \succeq\mu \kappa_B^8 \tau_y^4 \tau_z^4 K_y \log^2 N, \label{cond:2}\\
&\rho \sigma_B^2 \sqrt{n_y n_z} \succeq\frac{\kappa_B^2 \tau_y \tau_z K_y K_z}{\beta_y \beta_z} \log N, \label{cond:3}\\
&\rho \sigma_B^2 n_z \succeq\frac{\kappa_B^6 \tau_y^3 \tau_z^3 K_y K_z}{\beta_y \beta_z} \log N, \label{cond:4}\\
&n_y \succeq \frac{\kappa_B^4 \tau_y^2 \tau_z^2 K_y}{\beta_y}, \label{cond:5}\\
&\rho \min\{n_y,n_z\} \succeq\log N, \label{cond:6}\\
&n_y \gg \frac{\kappa_B^4 \tau_y^{5/2} \tau_z^2 \sqrt{\mu}\,K_y}{\sqrt{\beta_y}}, \label{cond:7}\\
&\rho \sigma_B^2 \sqrt{n_y n_z} \gg \frac{\kappa_B^2 \tau_y^{3/2} \tau_z \sqrt{\mu}\,K_y K_z \log N}{\sqrt{\beta_y}\,\beta_z}, \label{cond:8}\\
&\sqrt{\rho}\,\sigma_B \sqrt{n_z} \gg \frac{\kappa_B^3 \tau_y^2 \tau_z^{3/2}\sqrt{\mu K_yK_z \log N}}{\sqrt{\beta_z}}. \label{cond:9}
\end{align}
Then Algorithm~\ref{alg:SCDD} returns row labels satisfying \(\ell(\hat y,y)=0\) (i.e., exact recovery up to label permutation) with probability at least \(1-O(N^{-10})\).
\end{thm}

Theorem~\ref{thm:mainFinal} establishes a theoretical guarantee for exact community recovery in bipartite networks under mild conditions. The number of communities \(K_y\) is allowed to grow with the network size, and the communities on either side of the bipartite graph may be unbalanced, as reflected by the presence of the balance parameters \(\beta_y,\beta_z\) and the ratios \(\tau_y,\tau_z\) in the conditions. Moreover, the sparsity level \(\rho\) is required only to satisfy explicit lower bounds that scale with the network parameters. The theorem thus ensures that as long as the network is large and the signal meets these conditions, the SCDD algorithm recovers the true community structure exactly—even when communities are imbalanced and their number grows with the network.

\begin{rem}[On the conditions of Theorem~\ref{thm:mainFinal}]\label{remConScBM}
A careful inspection reveals that three of the nine conditions are redundant. Specifically, condition~(\ref{cond:5}) is implied by condition~(\ref{cond:7}), condition~(\ref{cond:3}) is implied by condition~(\ref{cond:8}), and condition~(\ref{cond:4}) is implied by condition~(\ref{cond:9}), as the latter conditions have strictly stronger right-hand sides up to constants. Thus, conditions~(\ref{cond:3}),~(\ref{cond:4}), and~(\ref{cond:5}) may be removed without weakening the theorem.

The remaining six conditions are not generally implied by one another. We retain all nine in the statement for two reasons. First, each condition corresponds to a distinct step in the proof, where it controls a specific error term in the perturbation analysis. Second, keeping the full set of conditions enhances readability by allowing the reader to directly match each assumption to its role in the argument, without the need to derive intermediate implications.
\end{rem}
\begin{rem}\label{rem:identifiability}
When \(K_y<K_z\), the model parameters remain identifiable (up to label permutations) as established by \cite{rohe2016co}. However, exact recovery of all column community labels is impossible because the row space of \(\Omega\) has dimension \(K_y<K_z\), so the observed signal only provides a \(K_y\)-dimensional projection of the column structure. Consequently, our theorem focuses on exact row recovery in this case. When \(K_y=K_z\), the column communities are also fully recoverable, and the following corollary provides the analogous exact recovery guarantee.
\end{rem}

\begin{cor}\label{cor:colRecovery}
Under ScBM with \(K_y=K_z\), suppose the conditions of Theorem~\ref{thm:mainFinal} hold after interchanging the roles of \(y\) and \(z\) (i.e., each occurrence of \(y\) is replaced by \(z\) and vice versa) (and consequently \(\mu_1\leftrightarrow\mu_2\)). Then Algorithm~\ref{alg:SCDD} returns column labels satisfying \(\ell(\hat z,z)=0\) with probability at least \(1-O(N^{-10})\).
\end{cor}
\subsection{Implications in the Balanced Regime}
In many practical scenarios, the communities are balanced and the parameters are of constant order. Under these
simplifications, the conditions in Theorem~\ref{thm:mainFinal} simplify considerably.
\begin{cor}\label{cor:balancedFinal}
Assume the balanced setting \(n_y \asymp n_z \asymp n\) (so \(N\asymp n\)), \(K_y=K_z=K\), and the quantities \(\kappa_B,\sigma_B,\beta_y,\beta_z,\tau_y,\tau_z,\mu\) are all \(O(1)\). Then, the conditions of Theorem~\ref{thm:mainFinal} simplify to
\[
K = o\!\left(\min\left\{\frac{n}{\log^2 N},\sqrt{\frac{\rho n}{\log N}}\right\}\right). 
\]

Consequently, we have
\begin{enumerate}
\item[(i)] \textbf{Sparse regime.} If \(\rho = c\log N/n\) for some absolute constant \(c>0\), then the sufficient conditions of Theorem~\ref{thm:mainFinal} fail (specifically, condition~(\ref{cond:9}) is violated). Consequently, Theorem~\ref{thm:mainFinal} does not apply in this regime.

\item[(ii)] \textbf{Dense regime.} If \(\rho=O(1)\), then exact row recovery holds provided
\[
K = o\!\left(\min\left\{\frac{n}{\log^2 N},\sqrt{\frac{n}{\log N}}\right\}\right). \label{eq:condFinalCorrect}
\]
Under this same condition, column recovery also holds whenever \(K_y=K_z\).
\end{enumerate}
\end{cor}

The assumptions \(\kappa_B=O(1)\), \(\sigma_B=O(1)\), \(\beta_y=O(1)\), \(\beta_z=O(1)\), \(\tau_y=O(1)\), \(\tau_z=O(1)\), and \(\mu=O(1)\) are regularity conditions. To interpret these quantities, recall that \(\sigma_B\) denotes the smallest singular value of the connectivity matrix \(B\); since \(B\in[0,1]^{K_y\times K_z}\), the bound \(\sigma_B=O(1)\) is automatic and, together with \(\sigma_B>0\), guarantees a non-vanishing signal. The condition number \(\kappa_B=\sigma_1(B)/\sigma_{K_y}(B)\) being bounded prevents \(B\) from being ill-conditioned. The parameters \(\beta_y=\frac{K_y n_{\min}^y}{n_y}\) and \(\beta_z=\frac{K_z n_{\min}^z}{n_z}\) measure the balance of community sizes, so \(\beta_y,\beta_z=O(1)\) means the smallest community is a constant fraction of the average size, ruling out extreme imbalance. The ratios \(\tau_y=\frac{n_{\max}^y}{n_{\min}^y}\) and \(\tau_z=\frac{n_{\max}^z}{n_{\min}^z}\) quantify the disparity between the largest and smallest communities, and \(\tau_y,\tau_z=O(1)\) excludes severe size heterogeneity. Finally, the condition \(\mu=O(1)\) ensures that the singular vectors of \(\Omega\) are sufficiently spread out and that the energy of \(\Omega\) is not concentrated on a few rows or columns. This incoherence assumption is mild and is guaranteed under weak conditions on \(B\) and the community sizes. For example, Lemma~\ref{lem:muBoundFinal} provides explicit conditions under which \(\mu=O(1)\) holds uniformly.

The balanced-setting corollary reveals a simplification: the nine conditions of Theorem~\ref{thm:mainFinal} reduce to two inequalities. These two conditions jointly characterize the trade-off among the number of communities \(K\), the network size \(n\), and the sparsity level \(\rho\): one limits the growth of \(K\) relative to \(n\), while the other ensures that the average degree is sufficiently large to overcome the noise induced by the community structure. When both hold, Theorem~\ref{thm:mainFinal} guarantees exact recovery.

The sparse regime in part (i) lies beyond the scope of our current proof technique, as the assumed sparsity level fails to satisfy the sufficient conditions required by our perturbation analysis. In particular, condition~(\ref{cond:9}) is violated, and consequently our theoretical guarantee does not apply in this regime. This does not necessarily reflect a fundamental limitation of the algorithm, but rather indicates that the current analytical framework requires stronger signal-to-noise conditions to establish exact recovery. Extending the guarantee to this sparser setting would require more refined perturbation arguments.

In contrast, the dense regime in part (ii) demonstrates that, under the mild regularity conditions stated above, our conditions reduce to the joint scaling 
\(K=o\!\left(\min\left\{\frac{n}{\log^2 N},\sqrt{\frac{n}{\log N}}\right\}\right)\),which follows directly from the conditions of Theorem~\ref{thm:mainFinal}. This condition allows the number of communities \(K\) to grow with the network size \(n\), but at a rate limited by the more restrictive of the two terms, typically the second one when \(n\) is large. This accommodates the common scenario in modern network analysis where the community structure is not fixed but rather increases with the number of nodes.

\begin{rem}[Polylogarithmically sparse regime]\label{rem:Polylogarithmically}
The sparse regime in Corollary~\ref{cor:balancedFinal}(i) corresponds to the logarithmic order \(\rho \asymp \log N/n\), which our current sufficient conditions do not cover. Nevertheless, our theory does cover the slightly denser but still very sparse regime
\[
\rho \asymp \frac{\log^{1+\varepsilon}N}{n},\qquad \varepsilon>0.
\]
Indeed, substituting this into the simplified condition of Corollary~\ref{cor:balancedFinal} gives
\[
\sqrt{\frac{\rho n}{\log N}} \asymp \log^{\varepsilon/2}N.
\]
Since \(n/\log^2 N \gg \log^{\varepsilon/2}N\) for large \(N\), the minimum is attained by the second term, so the condition becomes \(K = o(\log^{\varepsilon/2}N)\). Hence all sufficient conditions are satisfied, and exact recovery is guaranteed even when the average degree grows only polylogarithmically faster than the logarithmic order \(\log N\).
\end{rem}
\section{Degree-Corrected Stochastic Co-Blockmodels}\label{sec:DCScBM}

The stochastic co-blockmodel of Section~\ref{sec:SCBMrecovery} assumes that all nodes within the same community have identical expected degrees. While this simplifying assumption is useful for theoretical development, it is often violated in real-world bipartite networks, where nodes exhibit substantial heterogeneity in their connectivity levels—for instance, in a user-item network, some users are highly active while others interact only rarely. To address this limitation, we extend our analysis to the \emph{degree-corrected stochastic co-blockmodel} (DCScBM), originally proposed by \cite{karrer2011stochastic} for undirected graphs and adapted to bipartite settings by \cite{rohe2016co}. In this model, each node is endowed with an individual scale parameter that controls its overall connectivity. We demonstrate that the spectral methodology developed for the ScBM can be adapted to the DCScBM by incorporating a row-normalization step. Under mild conditions, we prove that the resulting algorithm achieves exact community recovery with high probability under DCScBM.

\subsection{Model and Algorithm}

We retain all notation from Section~\ref{sec:SCBMrecovery} for community sizes (\(n_k^y, n_l^z\)), balance parameters (\(\beta_y, \beta_z\)), imbalance ratios (\(\tau_y, \tau_z\)), the condition number of \(B\) (\(\kappa_B\)), the minimal singular value \(\sigma_B\), and the global network size \(N\). The incoherence parameters \(\mu_0, \mu_1, \mu_2, \mu\), and the singular value bounds \(\sigma_{\min}\) and \(\kappa_{\Omega}\) are also defined exactly as in Section~\ref{sec:SCBMrecovery}. We introduce only the new quantities needed to describe the degree heterogeneity. We now turn to a formal description of the model.

\begin{defin}\label{def:DCScBM}
A \emph{degree-corrected stochastic co-blockmodel} (DCScBM) with \(n_y\) row nodes, \(n_z\) column nodes, \(K_y\) row communities, and \(K_z\) column communities, where \(1\le K_y\le K_z\), is parameterized by:
\begin{itemize}
\item membership matrices \(Y\in\{0,1\}^{n_y\times K_y}\) and \(Z\in\{0,1\}^{n_z\times K_z}\), each row containing exactly one 1, with \(Y_{ik}=\mathbf{1}\{y_i=k\}\) and \(Z_{jl}=\mathbf{1}\{z_j=l\}\);
\item a connectivity matrix \(B\in[0,1]^{K_y\times K_z}\) satisfying \(\rank(B)=K_y\) and \(\sigma_B:=\sigma_{\min}(B)>0\);
\item degree heterogeneity parameters \(\theta_y\in\mathbb{R}^{n_y}_{>0}\) and \(\theta_z\in\mathbb{R}^{n_z}_{>0}\), with \(\Theta_y:=\diag(\theta_y)\) and \(\Theta_z:=\diag(\theta_z)\).
\end{itemize}
The observed adjacency matrix \(A\in\{0,1\}^{n_y\times n_z}\) has independent entries
\begin{align*}
A_{ij}\sim \mathrm{Bernoulli}\bigl(\theta_y(i) B_{y_i,z_j} \theta_z(j)\bigr),\quad 1\le i\le n_y,\;1\le j\le n_z,
\end{align*}
where the Bernoulli probabilities satisfy
\(\theta_y(i) B_{y_i,z_j} \theta_z(j)\in[0,1]\) for all \(i,j\). The expected adjacency matrix is \(\Omega:=\mathbb{E}[A]=\Theta_y Y B Z^\top \Theta_z\), and \(\rank(\Omega)=K_y\).
\end{defin}

To summarize the degree heterogeneity, we define
\begin{align*}
&\theta_{y,\min}:=\min_i \theta_y(i),\quad \theta_{y,\max}:=\max_i \theta_y(i),\quad
\theta_{z,\min}:=\min_j \theta_z(j),\quad \theta_{z,\max}:=\max_j \theta_z(j),\\
&\eta_y:=\frac{\theta_{y,\max}}{\theta_{y,\min}},\quad 
\eta_z:=\frac{\theta_{z,\max}}{\theta_{z,\min}},\quad
\theta_{\min}:=\theta_{y,\min}\theta_{z,\min},\quad
\theta_{\max}:=\theta_{y,\max}\theta_{z,\max},
\end{align*}
where the parameters \(\eta_y\) and \(\eta_z\) measure the spread of the degree parameters within each side of the bipartite graph.

Under DCScBM, although the singular vectors of \(\Omega\) are no longer constant on communities due to the multiplicative degree parameters, their row-normalized versions are. This observation, formalized in the following lemma, is the foundation of our spectral algorithm under DCScBM.
\begin{lem}\label{lem:SVstructureDC}
Under Definition~\ref{def:DCScBM}, let $\Omega=U\Sigma V^\top$ be the compact SVD.  Define the row-normalized matrix $U_*\in\mathbb{R}^{n_y\times K_y}$ by $(U_*)_{i,:}:=U_{i,:}/\|U_{i,:}\|_2$ for $1\le i\le n_y$. Then there exists an orthogonal matrix $Q$ such that $U_*=YQ$; consequently, for any distinct row communities $k\ne l$, we have $\|Q_{k,:}-Q_{l,:}\|_2=\sqrt{2}$. If additionally $K_y=K_z$, define $V_*$ by normalizing each row of $V$ to unit length. Then the same conclusion holds: there exists an orthogonal matrix $R$ such that $V_*=ZR$, and for distinct column communities $k\ne l$, $\|R_{k,:}-R_{l,:}\|_2=\sqrt{2}$.
\end{lem}

We now describe the algorithm. As in the ScBM, we first remove the diagonal bias from the Gram matrix; then we row-normalize the estimated singular vectors before clustering. We call this procedure \emph{Normalized Spectral Co‑clustering with Diagonal Deletion} (NSCDD), summarized in Algorithm~\ref{alg:NSCDD}.

\begin{algorithm}[!ht]
\caption{Normalized Spectral Co-clustering with Diagonal Deletion (NSCDD)}
\label{alg:NSCDD}
\begin{algorithmic}[1]
\Require Adjacency matrix \(A\in\{0,1\}^{n_y\times n_z}\), number of row communities \(K_y\) (and optionally \(K_z\) if column recovery is desired).
\Ensure Row labels \(\hat y\in[K_y]^{n_y}\).
\State Compute \(G_y:=\mathcal{P}_{\mathrm{off-diag}}(AA^\top)\in\mathbb{R}^{n_y\times n_y}\).
\State Let \(\widehat U\widehat\Lambda\widehat U^\top\) be the top-\(K_y\) eigen-decomposition of \(G_y\), where \(\widehat U\in\mathbb{R}^{n_y\times K_y}\) has orthonormal columns.
\State Row-normalize \(\widehat U\) to obtain \(\widehat U_*\): \((\widehat U_*)_{i,:}=\widehat U_{i,:}/\|\widehat U_{i,:}\|_2\).
\State Run \(k\)-means on the rows of \(\widehat U_*\) with \(K_y\) clusters and return the estimated labels \(\hat y\).
\State \textbf{Optional:} If \(K_y=K_z\), compute \(G_z:=\mathcal{P}_{\mathrm{off-diag}}(A^\top A)\), take its leading \(K_y\) eigenvectors \(\widehat V\),  row-normalize \(\widehat V\) to obtain \(\widehat V_*\), run \(k\)-means on the rows of \(\widehat V_*\) with \(K_y\) clusters, and return \(\hat z\).
\end{algorithmic}
\end{algorithm}
\subsection{Theoretical Guarantees}

We now present the main result of this section, which guarantees that NSCDD achieves exact row community recovery under DCScBM. 
\begin{thm}\label{thm:mainDC}
Under DCScBM, suppose the following conditions hold:
\begin{align}
& n_y n_z \succeq \mu^2 \kappa_B^8 \eta_y^8\eta_z^8 \tau_y^4\tau_z^4 K_y^2 \log^4 N, \label{cond:DC1}\\
& n_z \succeq \mu \kappa_B^8 \eta_y^8\eta_z^8 \tau_y^4\tau_z^4 K_y \log^2 N, \label{cond:DC2}\\
& \theta_{\min}\sigma_B^2 \sqrt{n_y n_z} \succeq
\frac{\kappa_B^2 \eta_y^3\eta_z^3 \tau_y \tau_z K_y K_z}{\beta_y \beta_z} \log N, \label{cond:DC3}\\
& \theta_{\min}\sigma_B^2 n_z \succeq
\frac{\kappa_B^6 \eta_y^7\eta_z^7 \tau_y^3\tau_z^3 K_y K_z}{\beta_y \beta_z} \log N, \label{cond:DC4}\\
& n_y \succeq \frac{\kappa_B^4 \eta_y^6\eta_z^4 \tau_y^2\tau_z^2 K_y}{\beta_y}, \label{cond:DC5}\\
& \theta_{\min} \min\{\sqrt{n_y n_z},\, n_z\} \succeq\frac{\log N}{\eta_y\eta_z}, \label{cond:DC6}\\
& n_y \gg \frac{\kappa_B^4 \eta_y^7 \eta_z^4 \tau_y^{5/2} \tau_z^2 \sqrt{\mu}\,K_y}{\sqrt{\beta_y}}, \label{cond:DC7}\\
& \theta_{\min}\sigma_B^2 \sqrt{n_y n_z}
\gg \frac{\kappa_B^2 \eta_y^4 \eta_z^3 \tau_y^{3/2} \tau_z \sqrt{\mu}\,K_y K_z \log N}{\beta_z\sqrt{\beta_y}}, \label{cond:DC8}\\
&\sigma_B \sqrt{\theta_{\min}n_z}
\gg \frac{\kappa_B^3 \eta_y^{9/2} \eta_z^{7/2} \tau_y^2 \tau_z^{3/2}\sqrt{\mu K_y K_z \log N}}{\sqrt{\beta_z}}. \label{cond:DC9}
\end{align}
Then Algorithm~\ref{alg:NSCDD} returns row labels satisfying \(\ell(\hat y,y)=0\) with probability at least \(1-O(N^{-10})\).
\end{thm}

The conditions in Theorem~\ref{thm:mainDC} quantify the interplay between sparsity, degree heterogeneity, community balance, and the number of communities. When the degree parameters are all equal, i.e., \(\Theta_y=\sqrt{\rho}I_{n_y}\) and \(\Theta_z=\sqrt{\rho}I_{n_z}\), DCScBM reduces to ScBM with sparsity \(\rho\), and the theorem reduces to Theorem~\ref{thm:mainFinal} exactly. Thus Theorem~\ref{thm:mainDC} directly generalizes Theorem~\ref{thm:mainFinal} and contains it as a special case.

\begin{rem}[On the conditions of Theorem~\ref{thm:mainDC}]
As noted in Remark~\ref{remConScBM}, three of the nine conditions are redundant: condition~(\ref{cond:DC5}) follows from~(\ref{cond:DC7}), condition~(\ref{cond:DC3}) follows from~(\ref{cond:DC8}), and condition~(\ref{cond:DC4}) follows from~(\ref{cond:DC9}). The remaining six conditions are not generally implied by one another. All nine are retained for the same two reasons given in Remark~\ref{remConScBM}: each condition controls a separate error term in the proof, and presenting the full set allows each assumption to be tied directly to its corresponding error term.
\end{rem}

The following corollary provides the symmetric guarantee when \(K_y=K_z\).

\begin{cor}\label{cor:colDC}
Under DCScBM, suppose \(K_y=K_z\). If the conditions of Theorem~\ref{thm:mainDC} hold with the roles of \(y\) and \(z\) interchanged (and \(\mu_1\leftrightarrow\mu_2\)), then Algorithm~\ref{alg:NSCDD} returns column labels satisfying \(\ell(\hat z,z)=0\) with probability at least \(1-O(N^{-10})\).
\end{cor}

Finally, to illustrate the implications of Theorem~\ref{thm:mainDC}, we specialize to the balanced regime. Under the same regularity assumptions as in Corollary~\ref{cor:balancedFinal} and with the additional mild assumptions that the degree heterogeneity parameters are bounded ($\eta_y,\eta_z=O(1)$) and that the minimum degree scale is of the same order as a reference sparsity parameter ($\theta_{\min}\asymp\rho$), the nine conditions of Theorem~\ref{thm:mainDC} simplify dramatically. The following corollary shows that the degree-corrected model does not alter the fundamental scaling requirements for exact recovery; the heterogeneity parameters only affect the constants.

\begin{cor}\label{cor:balancedDC}
Assume the balanced setting $n_y \asymp n_z \asymp n$ (so $N\asymp n$), $K_y=K_z=K$, and the quantities $\kappa_B,\sigma_B,\beta_y,\beta_z,\tau_y,\tau_z,\mu$ are all $O(1)$. Suppose additionally that $\eta_y,\eta_z=O(1)$ and $\theta_{\min}\asymp\rho$ for some reference sparsity parameter $\rho$. Then the conditions of Theorem~\ref{thm:mainDC} simplify to
\[
K = o\!\left(\min\left\{\frac{n}{\log^2 N},\sqrt{\frac{\theta_{\min}n}{\log N}}\right\}\right).
\]
Consequently:
\begin{enumerate}
\item[(i)] \textbf{Sparse regime.} If $\rho=c\log N/n$ for some absolute constant $c>0$, then the sufficient conditions of Theorem~\ref{thm:mainDC} fail, and exact recovery is not guaranteed by our theorem.

\item[(ii)] \textbf{Dense regime.} If $\rho=O(1)$, then $\theta_{\min}=O(1)$ as well, and exact row and column recovery holds provided
\[
K = o\!\left(\min\left\{\frac{n}{\log^2 N},\sqrt{\frac{n}{\log N}}\right\}\right).
\]
\end{enumerate}
\end{cor}
This corollary demonstrates that, in the balanced setting, the degree-corrected model does not alter the fundamental scaling requirements for exact recovery; the heterogeneity parameters $\eta_y,\eta_z$ and the minimum degree scale $\theta_{\min}$ only affect the constants. Thus, as long as the degree heterogeneity is not extreme and the effective average degree $\theta_{\min} n$ is sufficiently large, the algorithm achieves exact recovery with the same qualitative scaling as in the homogeneous case. Under the conditions stated in Lemma~\ref{lem:muBoundDC}, the incoherence parameter $\mu$ is guaranteed to be $O(1)$ under DCScBM, so the above simplification is justified within the framework of our theorem. Finally, as in the ScBM case, the polylogarithmically sparse regime
\(\theta_{\min} \asymp \log^{1+\varepsilon}N/n\) with \(\varepsilon>0\)
is also covered, provided that \(K=o(\log^{\varepsilon/2}N)\); this follows by the same substitution as in Remark~\ref{rem:Polylogarithmically}.
\section{Numerical Experiments}
\label{sec:num}

This section presents numerical studies to empirically validate the theoretical guarantees established in Theorems~\ref{thm:mainFinal} and~\ref{thm:mainDC}. For every experimental configuration, we generate \(100\) independent data sets and report the \emph{exact recovery proportion}, defined as the fraction of repetitions in which an algorithm simultaneously achieves perfect recovery of both row and column communities, i.e., \(\ell(\hat y,y)=0\) and \(\ell(\hat z,z)=0\).

We operate under the \emph{balanced setting}: \(K_y=K_z=K\), \(n_y\asymp n_z\asymp n\), and the quantities \(\mu,\kappa_B,\tau_y,\tau_z,\eta_y,\eta_z\) are all \(O(1)\). Furthermore, we consider the scenario where the two balance parameters are of the same order, i.e., \(\beta_y \asymp\beta_z\asymp\beta\). Under this setting, the nine conditions in Theorems~\ref{thm:mainFinal} and~\ref{thm:mainDC} simplify to the following two dominant requirements (see Corollaries~\ref{cor:balancedFinal} and~\ref{cor:balancedDC}). For the ScBM, the conditions are
\begin{align}
& n \gg K \log^2 N, \label{eq:C1}\\
& \rho n \gg \frac{K^2 \log N}{\beta^2}. \label{eq:C2}
\end{align}
For the DCScBM, \(\rho\) is replaced by \(\theta_{\min} = \min_{i,j}\theta_y(i)\theta_z(j)\), yielding the analogous conditions with \(\rho\) replaced by \(\theta_{\min}\) in condition \eqref{eq:C2}; condition \eqref{eq:C1} remains unchanged. Condition~\eqref{eq:C1} imposes only a mild restriction on \(K\), as it requires \(K \ll n / \log^2 N\), which is easily satisfied in practice for moderate to large \(n\). Condition~\eqref{eq:C2} captures the joint effect of sparsity and community imbalance and is typically the dominant constraint. In our experiments, we choose \(n\) and \(K\) such that condition~\eqref{eq:C1} is always satisfied, and focus on verifying condition~\eqref{eq:C2}, which determines the phase transition of exact recovery.
\subsection{Data Generation and Competing Methods}
For ScBM, the expected adjacency matrix is \(\Omega = \rho Y B Z^\top\), where \(Y\in\{0,1\}^{n_y\times K}\) and \(Z\in\{0,1\}^{n_z\times K}\) are membership matrices. To generate \(Y\), we choose an integer \(m_y\) with \(1 \le m_y \le \lfloor n_y/K \rfloor\), which serves as the minimum row community size \(n_{\min}^y = m_y\). We first assign \(m_y\) row nodes to each of the \(K\) row communities, then randomly distribute the remaining \(n_y - K m_y\) row nodes among the first \(K-1\) row communities. Thus the \(K\)-th row community has exactly \(m_y\) row nodes and all others have at least \(m_y\), so the minimum size is indeed \(m_y\). The matrix \(Z\) is generated analogously with \(m_z\) (\(1 \le m_z \le \lfloor n_z/K \rfloor\)), so \(n_{\min}^z = m_z\). The balance parameters are then
\[
\beta_y = \frac{K m_y}{n_y}, \qquad \beta_z = \frac{K m_z}{n_z}.
\]
Choosing \(m_y/n_y \asymp m_z/n_z\) yields \(\beta_y \asymp \beta_z\).  When \(m_y = O(n_y/K)\) and \(m_z = O(n_z/K)\), we have \(\tau_y,\tau_z = O(1)\), consistent with the balanced setting. In all experiments, the connectivity matrix is fixed as a non-symmetric matrix \(B\in[0,1]^{K\times K}\) with diagonal entries \(B_{kk}=0.9\), upper triangular entries \(B_{kl}=0.2\) for \(k<l\), and lower triangular entries \(B_{kl}=0.1\) for \(k>l\). This matrix has full rank and its entries are bounded, ensuring \(\sigma_K(B)=O(1)\) and \(\kappa_B=O(1)\), thereby satisfying the theoretical assumptions. For DCScBM, we set \(\Omega = \Theta_y Y B Z^\top \Theta_z\), where \(\Theta_y = \mathrm{diag}(\theta_y(1),\ldots,\theta_y(n_y))\), \(\Theta_z = \mathrm{diag}(\theta_z(1),\ldots,\theta_z(n_z))\), with
\[
\theta_y(i), \theta_z(j) \stackrel{\mathrm{i.i.d.}}{\sim} \mathrm{Uniform}(\frac{\rho^{1/2}}{3}, 1), \quad i=1,\ldots,n_y,\ j=1,\ldots,n_z.
\]
This gives \(\theta_{\min} = \min_{i,j}\theta_y(i)\theta_z(j) = \rho/9\) and \(\eta_y = \eta_z = 3\rho^{-1/2}\). In the dense regime where \(\rho=O(1)\), we have \(\theta_{\min}=O(\rho)\) and \(\eta_y,\eta_z=O(1)\), satisfying the theoretical assumptions for DCScBM. The same \(Y, B, Z\), and \(\rho\) are used in both models, with only the additional degree heterogeneity parameters \(\Theta_y,\Theta_z\) distinguishing DCScBM from ScBM. In each experimental setting described below, the parameters \(n_y, n_z, K, \rho, m_y\) (i.e., \(\beta_y\)), and \(m_z\) (i.e., \(\beta_z\)) will be independently specified according to the specific goal of that experiment.

We compare our proposed algorithms SCDD (designed for ScBM) and NSCDD (designed for DCScBM) against four representative spectral methods: \emph{DI-SIM}~\citep{rohe2016co}, a co-clustering method based on a regularized Laplacian that achieves consistency under DCScBM; \emph{D-SCORE}~\citep{wang2020spectral}, a spectral algorithm originally developed for directed networks, which we adapt to the bipartite setting and which also achieves consistency under DCScBM; and \emph{BiSC}~\citep{qing2023community}, a spectral co-clustering algorithm for bipartite networks that is consistent under ScBM, along with its degree-corrected counterpart \emph{nBiSC}~\citep{qing2023community}, which achieves consistency under DCScBM. In our ScBM experiments, BiSC serves as the primary baseline. In our DCScBM experiments, the primary baselines are DI-SIM, D-SCORE, and nBiSC, all of which enjoy consistency guarantees under DCScBM.
\subsection{Experimental Setup}

We conduct four experiments to verify the theoretical conditions. In Experiments 1–3, we set \(\beta_y=\beta_z=1\) by choosing \(m_y= \lfloor n_y/K \rfloor\) and \(m_z= \lfloor n_z/K \rfloor\). In Experiment 4, we vary \(\beta\) with \(\beta_y=\beta_z=\beta\). 

\paragraph{Experiment 1: Varying the Number of Communities \(K\)}
We fix \(n_y=800\), \(n_z=1200\), and \(\rho=0.3\), and let \(K\) range from 2 to 8.
Figure~\ref{fig:ex1} shows the exact recovery proportions under ScBM (left) and DCScBM (right).
For ScBM, all methods attain perfect recovery for \(K\le 7\), and the recovery rates decrease only marginally at \(K=8\).
For DCScBM, every algorithm achieves \(100\%\) recovery for \(K\le 4\), but performance drops substantially when \(K\ge 5\).
This earlier transition in DCScBM is expected because the degree heterogeneity reduces the effective signal to \(\theta_{\min}=\rho/9\), which is weaker than \(\rho\).
Overall, these observations are consistent with the theoretical condition~\eqref{eq:C2}: in ScBM the threshold is \(K \ll \sqrt{\rho n/\log N}\approx 6.22\), while in DCScBM the same condition with \(\rho\) replaced by \(\theta_{\min}\) yields a correspondingly smaller permissible \(K\), matching the observed phase shifts.

\begin{figure}[!htbp]
\centering
\resizebox{\columnwidth}{!}{
{\includegraphics[width=3\textwidth]{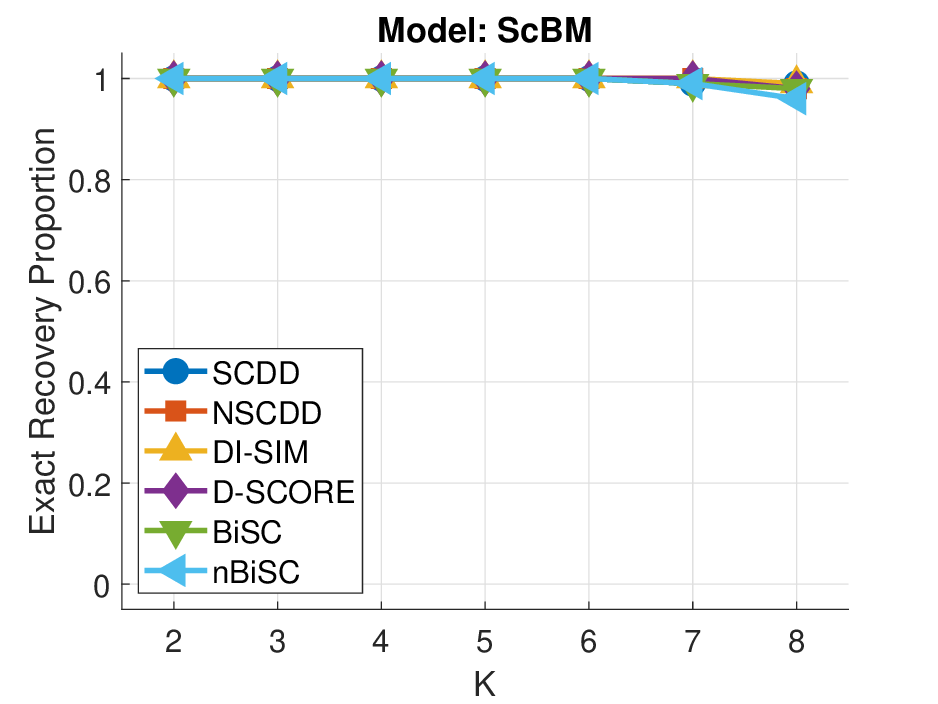}}
{\includegraphics[width=3\textwidth]{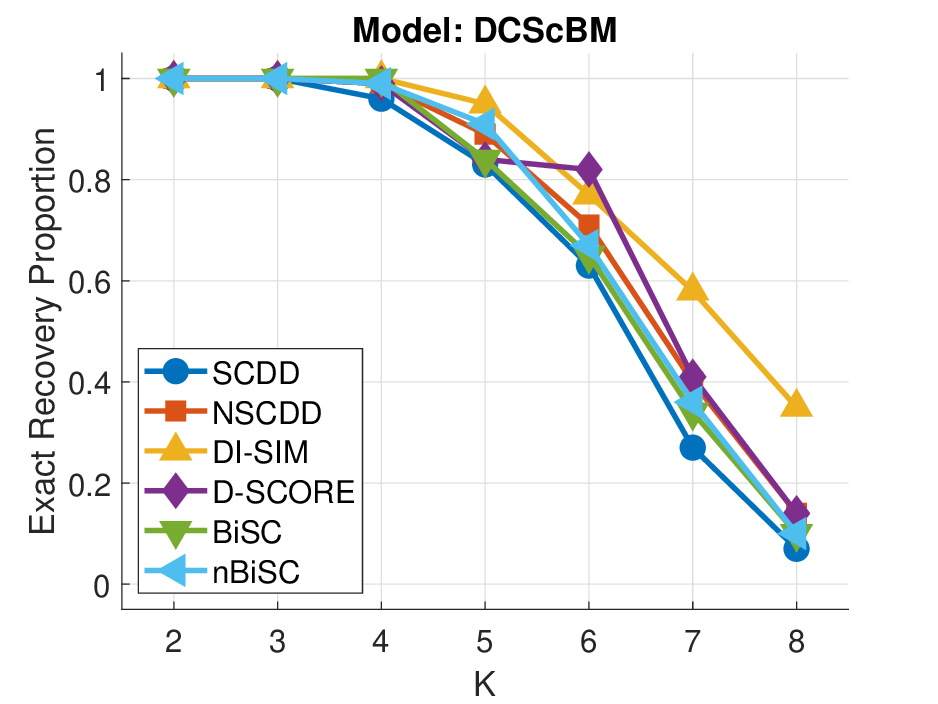}}
}
\caption{Numerical results of Experiment 1.}
\label{fig:ex1} 
\end{figure}

\paragraph{Experiment 2: Varying the Sparsity Parameter \(\rho\)}
We fix \(n_y=800\), \(n_z=1200\), and \(K=4\), and let \(\rho\) vary from \(0.05\) to \(0.50\). Figure~\ref{fig:ex2} reports the exact recovery proportions for ScBM (left) and DCScBM (right). In ScBM, all methods attain perfect recovery for \(\rho \ge 0.20\), while at \(\rho=0.15\) the rates are very close to but not exactly \(100\%\). For DCScBM, every algorithm achieves exact recovery for \(\rho \ge 0.35\), with near-perfect performance at \(\rho=0.30\) that falls slightly short. Condition~\eqref{eq:C2} for ScBM requires \(\rho \gg K^2\log N / n \approx 0.124\); the observed threshold at \(0.20\) is marginally above this asymptotic bound, as expected for finite samples. In DCScBM, the degree heterogeneity reduces the effective signal to \(\theta_{\min}=\rho/9\), so the same condition becomes \(\theta_{\min} \gg 0.124\), which necessitates a larger nominal \(\rho\). This explains the shift of the recovery threshold from \(0.20\) in ScBM to \(0.35\) in DCScBM, in full agreement with our theoretical predictions. 
\begin{figure}[!htbp]
\centering
\resizebox{\columnwidth}{!}{
{\includegraphics[width=3\textwidth]{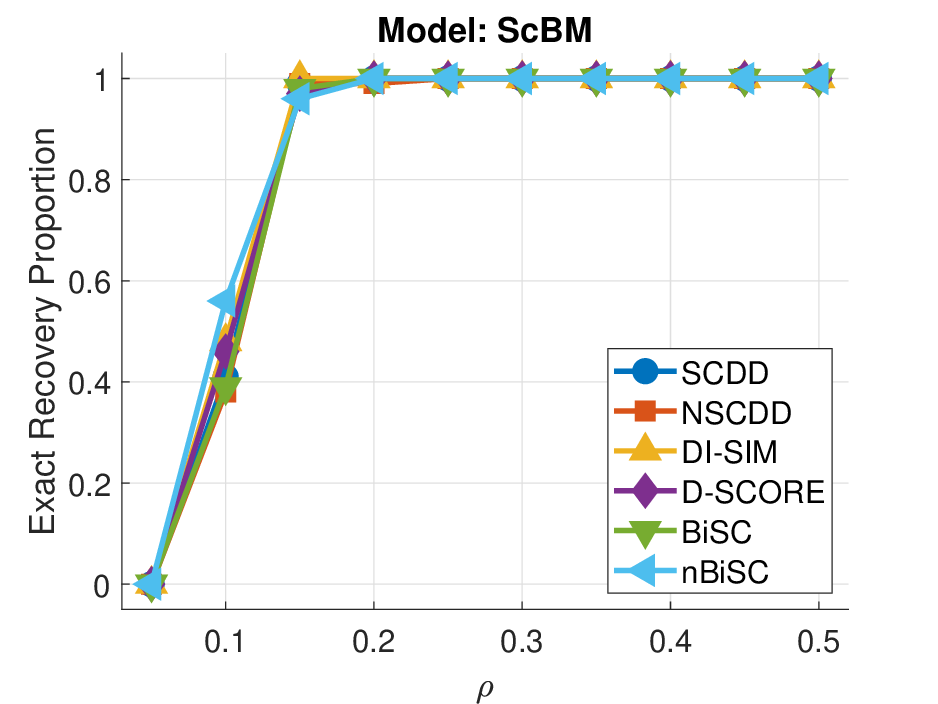}}
{\includegraphics[width=3\textwidth]{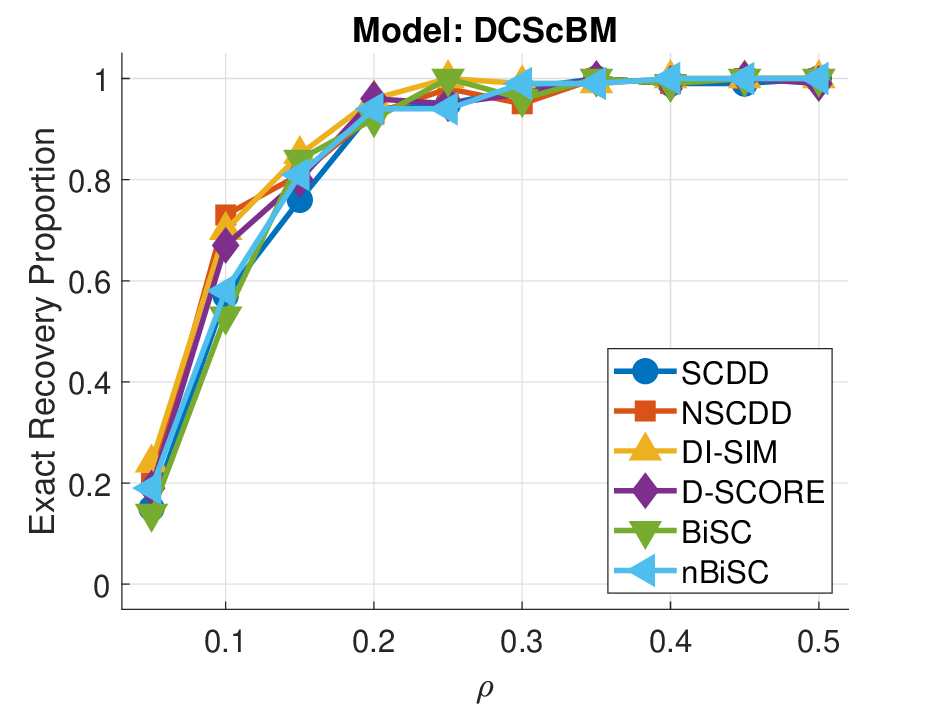}}
}
\caption{Numerical results of Experiment 2.}
\label{fig:ex2} 
\end{figure}

\paragraph{Experiment 3: Varying the Network Size \(n\)}
We fix \(K=4\) and \(\rho=0.5\), set \(n_y=n_z=n\), and vary \(n\) from \(100\) to \(1000\) with \(N=2n\); although \(n_y=n_z\), the membership matrices \(Y\) and \(Z\) are independently generated, so \(Y\ne Z\) with probability one. Figure~\ref{fig:ex3} reports the exact recovery proportions for ScBM (left) and DCScBM (right). In ScBM, all methods are very close to \(100\%\) recovery at \(n=200\) but do not attain exact recovery, and they achieve perfect recovery for all \(n\ge 300\). For DCScBM, every algorithm approaches exact recovery at \(n=600\) but does not quite reach it, and attains perfect recovery when \(n\ge 700\). Condition~\eqref{eq:C1} requires \(n \gg 4\log^2(2n)\), which is satisfied for \(n\ge 200\), while condition~\eqref{eq:C2} demands \(0.5n \gg 16\log(2n)\), i.e., \(n \gg 32\log(2n)\), placing the theoretical threshold between \(200\) and \(300\), consistent with the ScBM results. In DCScBM, the degree heterogeneity reduces the effective signal to \(\theta_{\min}=\rho/9\), which is weaker than \(\rho\); therefore, the same condition with \(\rho\) replaced by \(\theta_{\min}\) requires a larger \(n\), explaining why the recovery threshold shifts to between \(600\) and \(700\). 
\begin{figure}[!htbp]
\centering
\resizebox{\columnwidth}{!}{
{\includegraphics[width=3\textwidth]{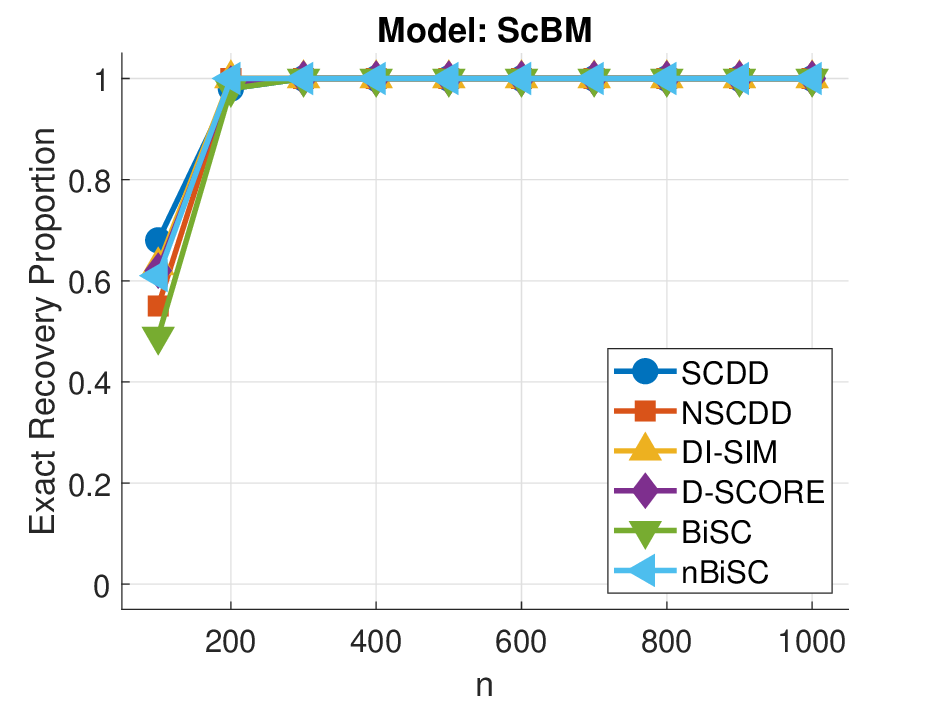}}
{\includegraphics[width=3\textwidth]{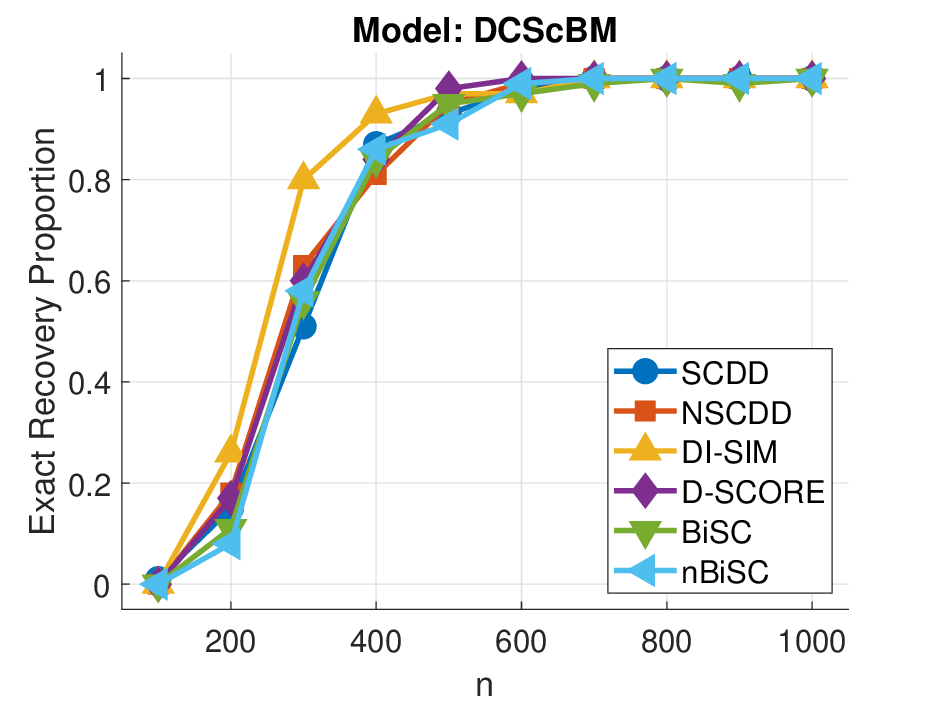}}
}
\caption{Numerical results of Experiment 3.}
\label{fig:ex3} 
\end{figure}

\paragraph{Experiment 4: Effects of Community Imbalance (Varying \(\beta\))}
We now investigate the effect of community imbalance by varying \(\beta\), with \(\beta_y=\beta_z=\beta\). We set \(n_y=2000\) and \(n_z=3000\) (so \(n_z=1.5 n_y\)), with \(K=4\) and \(\rho=0.2\). To achieve \(\beta_y=\beta_z=\beta\), we choose the minimum community sizes as \(m_y = 500\beta\) and \(m_z = 1.5 m_y = 750\beta\). Condition~\eqref{eq:C1} is clearly satisfied since \(n = \sqrt{2000 \times 3000} \approx 2449 \gg 4\log^2(5000) \approx 290\). Condition~\eqref{eq:C2} requires
\[
\beta^2 \gg \frac{K^2 \log N}{\rho \sqrt{n_y n_z}}.
\]
Substituting \(K=4,\ \rho=0.2,\ n_y=2000,\ n_z=3000,\ N=5000\), we obtain
\[
\beta \gg 
\sqrt{
\frac{4^2 \log(5000)}{0.2 \sqrt{2000 \times 3000}}
}
\approx 0.5274.
\]
We therefore choose \(\beta\) from \(\{0.1, 0.2, \ldots, 1.0\}\), corresponding to \(m_y = \{50, 100, \ldots, 500\}\) and \(m_z = \{75, 150, \ldots, 750\}\). Figure~\ref{fig:ex4} reports the exact recovery rates for ScBM (left) and DCScBM (right). In ScBM, all methods are very close to \(100\%\) recovery at \(\beta=0.5\) but do not attain exact recovery, and they achieve perfect recovery for all \(\beta\ge0.6\). In DCScBM, NSCDD, DI-SIM, and nBiSC attain near-perfect recovery at \(\beta=0.4\) yet fall slightly short, and reach exact recovery when \(\beta\ge0.5\); by contrast, SCDD, D-SCORE, and BiSC require larger \(\beta\) values to achieve perfect recovery. This overall later phase transition in DCScBM is expected because the degree heterogeneity reduces the effective signal to \(\theta_{\min}=\rho/9\), which is weaker than \(\rho\), so the same condition with \(\rho\) replaced by \(\theta_{\min}\) necessitates a larger \(\beta\), matching the theoretical prediction. 
\begin{figure}[!htbp]
\centering
\resizebox{\columnwidth}{!}{
{\includegraphics[width=3\textwidth]{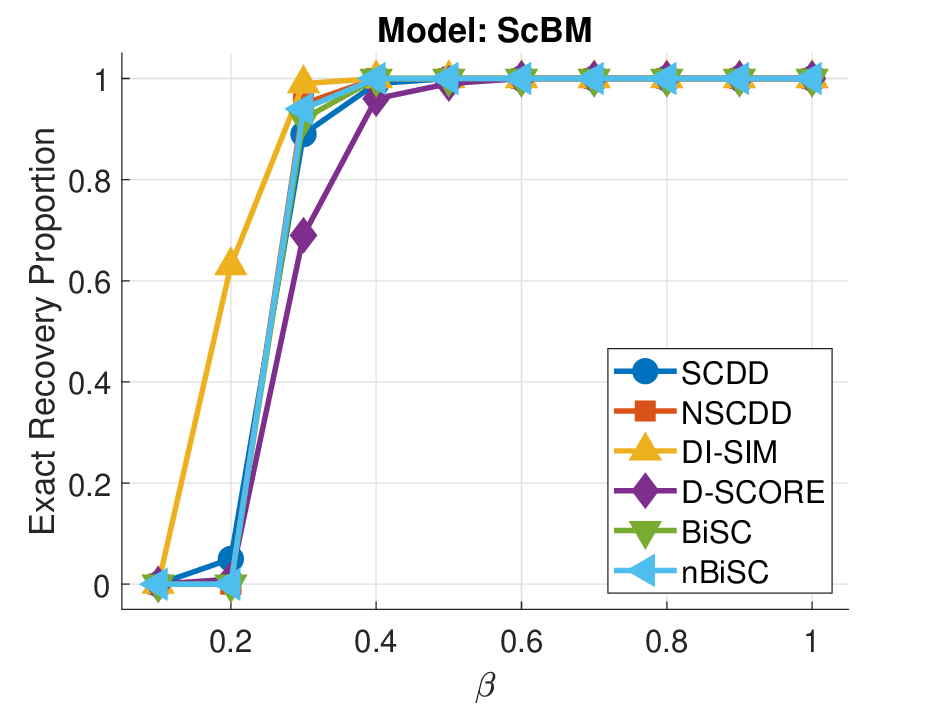}}
{\includegraphics[width=3\textwidth]{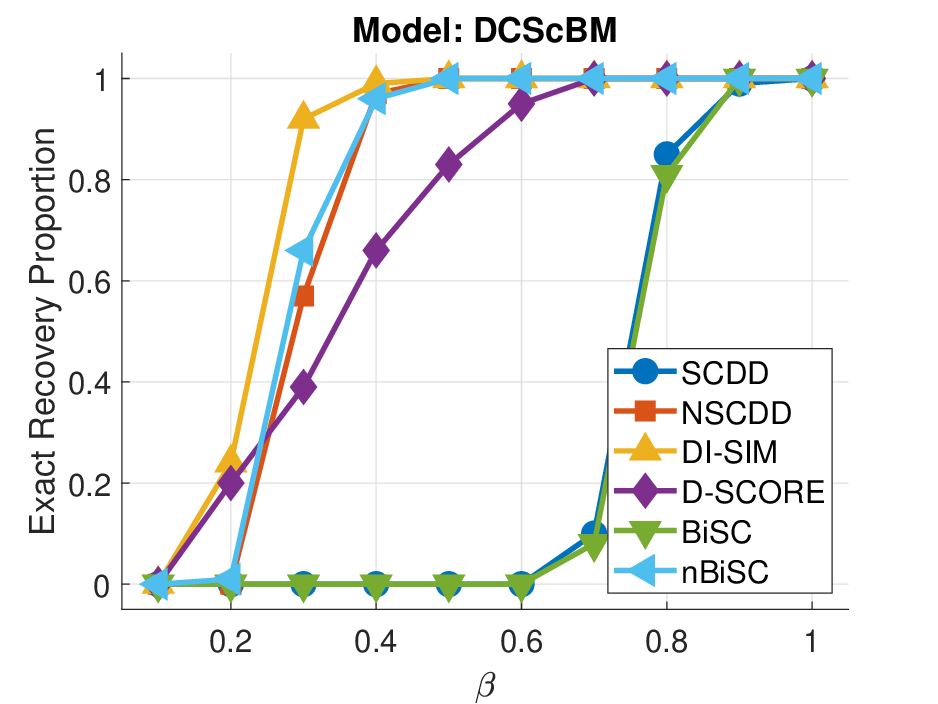}}
}
\caption{Numerical results of Experiment 4.}
\label{fig:ex4} 
\end{figure}
\section{Conclusion}\label{sec:conclusion}
In this paper, we have developed spectral methods for exact community recovery in bipartite networks under the stochastic co-blockmodel and its degree-corrected extension. The proposed algorithms, SCDD and NSCDD, are easy to implement: they construct a diagonal-deleted Gram matrix from the observed bi-adjacency matrix, extract its leading eigenvectors, and then apply $k$-means (with an optional row-normalization step in the degree-corrected case) to recover the node labels on both sides of the bipartition. Our main theoretical results give explicit conditions under which these algorithms recover the true row and column communities exactly with high probability. These conditions are stated directly in terms of the network sparsity, the number of communities, the condition number of the connectivity matrix, and the balance of community sizes, and they allow the number of communities to grow with the network size. In the balanced and homogeneous setting, the conditions reduce to two inequalities that reveal the trade-off between the number of communities, the network size, and the sparsity level. The analysis also covers the degree-corrected model, where each node has an individual scale parameter, and shows that the same spectral approach, with row normalization, retains the exact recovery guarantee. Numerical experiments support our theoretical findings.

Beyond the present scope, several directions remain for future investigation. It would be of interest to extend the analysis to weighted bipartite networks, where edge weights carry additional information beyond binary interactions, and to multi-layer bipartite networks, where multiple relations or temporal snapshots are observed simultaneously. Our current analysis requires the average degree to grow at least polylogarithmically in the network size, and whether the proposed algorithm can achieve exact recovery under sparser conditions remains to be examined; this would likely require a different proof technique. Our balanced-setting corollary gives a rate for the growth of the number of communities, and it would be interesting to determine whether this rate can be improved. A related question is whether one can achieve exact recovery by working directly with the singular vectors of the bi-adjacency matrix, bypassing the Gram matrix; this would be computationally more efficient but would require tools for entrywise eigenvector analysis of non-square random matrices. A particularly important question concerns the role of Laplacian normalization in exact community recovery under ScBM and its degree-corrected variant. \cite{rohe2016co} introduced these models and designed a spectral method based on the normalized Laplacian, establishing consistency in the sense of vanishing misclassification proportion. A natural next step is to ask whether Laplacian-based spectral methods can be strengthened to achieve exact recovery with zero misclassification error. This question remains unresolved. The main difficulty is that the Laplacian-based approach lacks the diagonal-deletion step that removes bias in our analysis. Moreover, the normalization factors introduce complex statistical dependencies that make the additive noise framework of \cite{cai2021subspace} inapplicable. Consequently, the existing perturbation theory cannot be directly used, and a new analytical framework is needed. Resolving this question would clarify whether Laplacian-based spectral methods can achieve exact recovery under the ScBM and DCScBM models, thereby bridging the Gram-matrix-based approach employed in this work and the Laplacian-based method of \cite{rohe2016co}. This would likely require a new perturbation theory for random matrices with multiplicative normalization, which remains to be developed.
\appendix

\section{Technical proofs under ScBM}\label{app:A}

\subsection{Proof of Lemma~\ref{lem:SVstructure}}\label{app:A.1}
\begin{proof}
Since every community is non-empty by definition, the diagonal matrices
\begin{align*}
D_y := Y^\top Y = \operatorname{diag}(n_1^y,\dots,n_{K_y}^y) \in \mathbb{R}^{K_y \times K_y}, 
\end{align*}
and
\begin{align*}
D_z := Z^\top Z = \operatorname{diag}(n_1^z,\dots,n_{K_z}^z) \in \mathbb{R}^{K_z \times K_z},
\end{align*}
are both positive definite and hence invertible.

We first establish the structural representation of the singular vectors. Since $\Omega = \rho Y B Z^\top$, the column space of $\Omega$ is a subspace of the column space of $Y$. Because $\operatorname{rank}(\Omega) = K_y$ (as $\operatorname{rank}(B) = K_y$ and both $Y$ and $Z$ have full column rank) and $Y \in \mathbb{R}^{n_y \times K_y}$ has full column rank, the column space of $\Omega$ coincides with the column space of $Y$. Hence the left singular vectors $U$, which form an orthonormal basis for the column space of $\Omega$, can be expressed as a linear combination of the columns of $Y$. Therefore, there exists a matrix $X \in \mathbb{R}^{K_y \times K_y}$ such that
\begin{align*}
U = Y X. 
\end{align*}

Analogously, the row space of $\Omega$ (equivalently, the column space of $\Omega^\top$) is contained in the column space of $Z$. Because the row space of $\Omega$ has dimension $K_y$, which is at most $K_z$, there exists a matrix $W \in \mathbb{R}^{K_z \times K_y}$ satisfying
\begin{align*}
V = Z W.
\end{align*}

Moreover, since $V \in \mathbb{R}^{n_z \times K_y}$ has full column rank $K_y$ (as $V^\top V = I_{K_y}$) and $Z \in \mathbb{R}^{n_z \times K_z}$ has full column rank, the representation $V = Z W$ implies that $W$ must have full column rank $K_y$; otherwise, $\operatorname{rank}(V) \le \operatorname{rank}(W) < K_y$, a contradiction. Thus $W$ is full column rank. This establishes the existence of matrices $X$ and $W$.

We now prove that $X$ is invertible and derive the explicit distance formula for its rows. From the orthonormality of $U$ and the representation $U = Y X$, we obtain
\begin{align}
I_{K_y} = U^\top U = (Y X)^\top (Y X) = X^\top (Y^\top Y) X = X^\top D_y X. \label{eq:XorthScBM}
\end{align}

Taking determinants on both sides of Equation \eqref{eq:XorthScBM} yields
\begin{align*}
1 = \det(I_{K_y}) = \det(X^\top D_y X) = (\det X)^2 \det(D_y). 
\end{align*}

Because $\det(D_y) = \prod_{k=1}^{K_y} n_k^y > 0$, it follows that $(\det (X))^2 = 1 / \det(D_y) \neq 0$, so $\det (X) \neq 0$; hence $X$ is invertible.

Next, we derive the precise relation between $X$ and $D_y$. Since $X$ is invertible, we may invert both sides of Equation \eqref{eq:XorthScBM} to obtain
\begin{align}
X^{-1} D_y^{-1} X^{-\top} = I_{K_y}, \label{eq:XinvD}
\end{align}
where $X^{-\top} := (X^{-1})^\top$. Left-multiplying Equation \eqref{eq:XinvD} by $X$ and right-multiplying by $X^\top$ yields
\begin{align*}
X X^{-1} D_y^{-1} X^{-\top} X^\top = X I_{K_y} X^\top,
\end{align*}
which simplifies to $D_y^{-1} = X X^\top$. This identity is the key to the row-distance calculation.

For any distinct indices $k, l \in \{1,\dots,K_y\}$, let $e_k \in \mathbb{R}^{K_y}$ denote the $k$-th standard basis vector. We have
\begin{align*}
\|X_{k,:} - X_{l,:}\|_2^2= (e_k - e_l)^\top X X^\top (e_k - e_l)= (e_k - e_l)^\top D_y^{-1} (e_k - e_l)= \frac{1}{n_k^y} + \frac{1}{n_l^y}.
\end{align*}

Since both $n_k^y$ and $n_l^y$ are strictly positive, the right-hand side is positive for $k \neq l$, and taking square roots yields  \(\|X_{k,:}-X_{l,:}\|_2 = \sqrt{\frac{1}{n_k^y}+\frac{1}{n_l^y}}\).

For the right singular vectors when $K_y=K_z$, $W$ is square. Repeating the above analysis completes the proof.
\end{proof}

\subsection{Proof of Theorem~\ref{thm:mainFinal}}\label{app:B}
\begin{proof}
We apply Theorem~1 of \cite{cai2021subspace} to the bi-adjacency matrix \(A\), viewed as a noisy observation of \(\Omega\). In the notation of that theorem, we set \(d_1=n_y\), \(d_2=n_z\), \(r=K_y\), and \(p_{\rm samp}=1\) (full observation). To invoke the theorem, we need to verify its three conditions in Equation~(15) as well as its Assumption~2.

For this purpose, we shall repeatedly use the following three estimates, which follow directly from Lemmas~\ref{lem:sigmaMinLower}, \ref{lem:kappaUpperScBM}, and \ref{lem:mu1ScBM}, respectively:
\[
\sigma_{\min}\ge \rho\sigma_B\sqrt{\frac{\beta_y\beta_z n_y n_z}{K_y K_z}},\quad
\kappa_{\Omega}\le \kappa_B\sqrt{\tau_y\tau_z},\quad
\mu_1=\frac{1}{\beta_y}.
\]

Together with the explicit assumptions in Theorem~\ref{thm:mainFinal}, these bounds allow us to verify all required conditions. The detailed verification is as follows.
\begin{itemize}
\item From conditions~(\ref{cond:1}) and~(\ref{cond:2}), together with \(\kappa_{\Omega}\le \kappa_B\sqrt{\tau_y\tau_z}\), we have
\[
n_y n_z \succeq \mu^2(\kappa_B\sqrt{\tau_y\tau_z})^8 K_y^2 \log^4 N
       \ge \mu^2 \kappa_{\Omega}^8 K_y^2 \log^4 N,
\]
and
\[
n_z \succeq \mu \kappa_B^8 \tau_y^4 \tau_z^4 K_y \log^2 N
      \ge \mu \kappa_{\Omega}^8 K_y \log^2 N.
\]
These directly imply the first condition in Equation~(15) of \cite{cai2021subspace}:
\[
1 \succeq \max\left\{
\frac{\mu \kappa_{\Omega}^4 K_y \log^2 N}{\sqrt{n_y n_z}},
\frac{\mu \kappa_{\Omega}^8 K_y \log^2 N}{n_z}
\right\}.
\]
\item Let \(M:=A-\Omega\). The entries of \(M\) are independent, zero-mean, bounded in absolute value by \(1\), with variance at most \(\rho\). By conditions~(\ref{cond:3}) and~(\ref{cond:4}), together with the lower bound on \(\sigma_{\min}\) and the bound on \(\kappa_{\Omega}\), we have
\[
\sigma_{\min}\ge \rho\sigma_B\sqrt{\frac{\beta_y\beta_z n_y n_z}{K_y K_z}}
\succeq \kappa_B\sqrt{\tau_y\tau_z}\,(n_y n_z)^{1/4}\sqrt{\rho\log N}
\ge \kappa_{\Omega}(n_y n_z)^{1/4}\sqrt{\rho\log N},
\]
and
\[
\sigma_{\min}\ge \rho\sigma_B\sqrt{\frac{\beta_y\beta_z n_y n_z}{K_y K_z}}
\succeq \kappa_B^3(\tau_y\tau_z)^{3/2}\sqrt{\rho n_y\log N}
\ge \kappa_{\Omega}^3\sqrt{\rho n_y\log N}.
\]
These directly yield the second condition in Equation~(15) of \cite{cai2021subspace}:
\[
\frac{\sqrt{\rho}}{\sigma_{\min}}
\preceq \min\left\{
\frac{1}{\kappa_{\Omega}(n_y n_z)^{1/4}\sqrt{\log N}},
\frac{1}{\kappa_{\Omega}^3\sqrt{n_y\log N}}
\right\}.
\]
\item From condition~(\ref{cond:5}), together with \(\mu_1=1/\beta_y\) and \(\kappa_{\Omega}\le \kappa_B\sqrt{\tau_y\tau_z}\), we have
\[
n_y \succeq \frac{\kappa_B^4 \tau_y^2 \tau_z^2 K_y}{\beta_y}
= \mu_1 \kappa_B^4 \tau_y^2 \tau_z^2 K_y
\ge \mu_1 \kappa_{\Omega}^4 K_y.
\]
This directly yields the third condition in Equation~(15) of \cite{cai2021subspace}:
\[
K_y \preceq \frac{n_y}{\mu_1 \kappa_{\Omega}^4}.
\]
\end{itemize}

Condition~(\ref{cond:6}) ensures that Assumption~2 of \cite{cai2021subspace} is satisfied. Having verified all three conditions in Equation~(15) together with Assumption~2, we now apply Theorem~1 of \cite{cai2021subspace}. Consequently, with probability at least \(1-O(N^{-10})\), there exists an orthogonal matrix \(\mathcal{O}\in\mathbb{R}^{K_y\times K_y}\) such that
\[
\|\widehat U\mathcal{O} - U\|_{2,\infty}
\preceq \kappa_{\Omega}^2 \sqrt{\frac{\mu K_y}{n_y}}\, \mathcal{E}_{\rm gen},
\]
where
\[
\mathcal{E}_{\rm gen}
= \frac{\mu_1\kappa_{\Omega}^2 K_y}{n_y}
+ \frac{\rho\sqrt{n_y n_z}\log N}{\sigma_{\min}^2}
+ \frac{\sqrt{\rho}\,\kappa_{\Omega}\sqrt{n_y\log N}}{\sigma_{\min}}.
\]

For the diagonal bias term, by condition (\ref{cond:7}),  we have
\begin{align*}
\kappa^2_{\Omega} \sqrt{\frac{\mu K_y}{n_y}} \cdot \frac{\mu_1\kappa^2_{\Omega} K_y}{n_y}
&=\kappa^4_{\Omega} \sqrt{\frac{\mu K_y}{n_y}} \cdot \frac{K_y}{\beta_yn_y}\leq\kappa^4_{B}\tau^2_y\tau^2_z \sqrt{\frac{\mu K_y}{n_y}} \cdot \frac{K_y}{\beta_yn_y}\ll\sqrt{\frac{K_y}{\tau_y\beta_yn_y}}. 
\end{align*}

For the quadratic noise term, by condition (\ref{cond:8}), we have
\begin{align*}
\kappa^2_{\Omega} \sqrt{\frac{\mu K_y}{n_y}} \cdot \frac{\rho\sqrt{n_y n_z}\log N}{\sigma_{\min}^2}
&\leq \kappa^2_B\tau_y\tau_z\sqrt{\frac{\mu K_y}{n_y}} \cdot \frac{\rho K_yK_z\sqrt{n_y n_z}\log N}{\rho^2\sigma^2_B\beta_y\beta_zn_yn_z}\ll\sqrt{\frac{K_y}{\tau_y\beta_yn_y}}.
\end{align*}

For the linear noise term, by condition (\ref{cond:9}), we have
\begin{align*}
\kappa^2_{\Omega} \sqrt{\frac{\mu K_y}{n_y}} \cdot \frac{\sqrt{\rho}\,\kappa_{\Omega}\sqrt{n_y\log N}}{\sigma_{\min}}
&\leq \kappa^3_{B}(\tau_y\tau_z)^{\frac{3}{2}} \sqrt{\frac{\mu K_y}{n_y}} \cdot \frac{\sqrt{\rho}\sqrt{K_yK_zn_y\log N}}{\rho\sigma_B\sqrt{\beta_y\beta_zn_yn_z}}\ll\sqrt{\frac{K_y}{\tau_y\beta_yn_y}}.
\end{align*}

Since $\delta_{r,\mathrm{min}}=\mathrm{min}_{k\neq l}\|X_{k,:}-X_{l,:}\|_{2}\geq\sqrt{\frac{2}{n^{y}_{\max}}}=\sqrt{\frac{2}{\tau_{y} n^{y}_{\min}}}=\sqrt{\frac{2K_y}{\tau_{y}\beta_{y}n_{y}}}$ by Lemma~\ref{lem:SVstructure}, we obtain
\begin{align*}
\|\widehat U\mathcal{O} - U\|_{2,\infty}
= o(\delta_{r,\mathrm{min}}). 
\end{align*}

Hence every row of \(\widehat U\mathcal O\) is strictly closer to its own class centroid than to any other. Therefore, \(k\)-means on \(\widehat U\mathcal O\) recovers the true row labels exactly. Since \(\mathcal O\) is orthogonal, it only permutes labels, so \(k\)-means on \(\widehat U\) gives the same clustering up to permutation; thus \(\ell(\hat y,y)=0\). 
\end{proof}
\subsection{Proof of Corollary \ref{cor:colRecovery}}
\begin{proof}
Apply Theorem~1 of \cite{cai2021subspace} to the matrix \(G_z=\mathcal{P}_{\mathrm{off-diag}}(A^\top A)\), viewed as a noisy observation of \(\bar G_z=\mathcal{P}_{\mathrm{off-diag}}(\Omega^\top\Omega)\). Since \(K_y=K_z\), the model and algorithm are symmetric in the two node types: the left singular vectors of \(\Omega^\top\) directly encode the column community structure. Hence, the entire verification of the three conditions in Equation~(15) and Assumption~2 of \cite{cai2021subspace} is identical to that in the proof of Theorem~\ref{thm:mainFinal} after interchanging \(y\leftrightarrow z\) and \(\mu_1\leftrightarrow\mu_2\). The same argument yields an approximation error for the column eigenvectors that is \(o(\delta_{c,\min})\), where \(\delta_{c,\min}\) is the minimum distance between distinct column community centers (provided by Lemma~\ref{lem:SVstructure}). The nearest-centroid argument then guarantees exact column recovery, i.e., \(\ell(\hat z,z)=0\).
\end{proof}
\subsection{Proof of Corollary~\ref{cor:balancedFinal}}\label{app:D}
\begin{proof}
Under the balanced setting and the assumed \(O(1)\) regularity conditions, the nine conditions of Theorem~\ref{thm:mainFinal} reduce to the two conditions \(K\log^2 N=o(n)\) and \(\rho n\gg K^2\log N\). Indeed, conditions~(\ref{cond:1}),~(\ref{cond:2}),~(\ref{cond:5}) and~(\ref{cond:7}) all reduce to \(K\log^2 N=o(n)\); conditions~(\ref{cond:3}),~(\ref{cond:4}),~(\ref{cond:8}) and~(\ref{cond:9}) reduce to \(\rho n\gg K^2\log N\) (with the last being implied by it since \(K\ge1\)); and condition~(\ref{cond:6}) is implied by \(\rho n\succeq\log N\), which follows from \(\rho n\gg K^2\log N\).

For part (i), substituting \(\rho=c\log N/n\) into \(\rho n\gg K^2\log N\) yields \(c\log N\gg K^2\log N\), i.e., \(K=o(1)\), which is impossible for growing \(K\). Hence the theorem's conditions fail.

For part (ii), when \(\rho=O(1)\), the condition \(\rho n\gg K^2\log N\) becomes \(n\gg K^2\log N\), i.e., \(K=o(\sqrt{n/\log N})\). Together with \(K\log^2 N=o(n)\), this yields \(K=o(\min\{n/\log^2 N,\sqrt{n/\log N}\})\). The theorem then guarantees exact row recovery, and column recovery follows from Corollary~\ref{cor:colRecovery} when \(K_y=K_z\).
\end{proof}
\subsection{Auxiliary lemmas under ScBM}\label{app:E}
\begin{lem}\label{lem:sigmaMinLower}
Under ScBM, we have $\sigma_{\min}\ge \rho\,\sigma_B\,\sqrt{\frac{\beta_{y}\beta_{z}n_y n_z}{K_y K_z}}$. 
\end{lem}
\begin{proof}
Since \(\Omega=\rho YBZ^\top\) under ScBM, we have $\sigma_{\min}\ge \rho\sigma_B\sqrt{n_{\min}^y n_{\min}^z}$ by basic algebra. Thus, the balanced bound follows by substituting $n_{\min}^y=\beta_{y} n_y/K_y$ and $n_{\min}^z=\beta_{z} n_z/K_z$.
\end{proof}
\begin{lem}\label{lem:kappaUpperScBM}
Under ScBM, we have $\kappa_{\Omega} := \frac{\sigma_1(\Omega)}{\sigma_{K_y}(\Omega)}
\;\le\; \kappa_B \sqrt{\tau_y\,\tau_z}$.
\end{lem}
\begin{proof}
Define the diagonal matrices of class sizes:
\[
D_y := Y^\top Y = \operatorname{diag}(n_1^y,\dots,n_{K_y}^y),\quad
D_z := Z^\top Z = \operatorname{diag}(n_1^z,\dots,n_{K_z}^z).
\]

Let
\[
\widetilde{Y} := Y D_y^{-1/2},\quad
\widetilde{Z} := Z D_z^{-1/2},\quad
\widetilde{B} := D_y^{1/2} B D_z^{1/2}.
\]

Then $\widetilde{Y}^\top \widetilde{Y} = I_{K_y}$ and $\widetilde{Z}^\top \widetilde{Z} = I_{K_z}$, and we have the key factorization
\[
\Omega = \rho \, Y B Z^\top
      = \rho \, (Y D_y^{-1/2}) (D_y^{1/2} B D_z^{1/2}) (D_z^{-1/2} Z^\top)
      = \rho \,\widetilde{Y} \widetilde{B} \widetilde{Z}^\top.
\]

We first verify that $\Omega$ and $\widetilde{B}$ share the same non-zero singular values up to the scalar $\rho$. Indeed,
\[
\Omega \Omega^\top = \rho^2 \widetilde{Y} \widetilde{B} \widetilde{Z}^\top \widetilde{Z} \widetilde{B}^\top \widetilde{Y}^\top = \rho^2 \widetilde{Y} (\widetilde{B} \widetilde{B}^\top) \widetilde{Y}^\top.
\]

Since $\widetilde{Y}^\top \widetilde{Y}=I_{K_y}$, the matrix $\widetilde{Y} (\widetilde{B} \widetilde{B}^\top) \widetilde{Y}^\top$ has the same non-zero eigenvalues as $\widetilde{B} \widetilde{B}^\top$. Hence the non-zero singular values of $\Omega$ are exactly $\rho$ times those of $\widetilde{B}$. Since $\operatorname{rank}(\Omega)=\operatorname{rank}(\widetilde{B})=K_y$, the $K_y$-th singular value is the smallest positive singular value in both cases, and the common factor $\rho$ cancels in the ratio. Therefore, we have
\[
\kappa_{\Omega} = \kappa_{\widetilde{B}} := \frac{\sigma_1(\widetilde{B})}{\sigma_{K_y}(\widetilde{B})}.
\]

It remains to bound $\kappa(\widetilde{B})$. Let $P := D_y^{1/2}$ and $Q := D_z^{1/2}$, so that $\widetilde{B} = P B Q$. For the largest singular value, the sub-multiplicativity of the spectral norm gives
\begin{align}\label{ScBMBBL}
\sigma_1(\widetilde{B}) \le \|P\|\,\|B\|\,\|Q\|
= \sqrt{n_{\max}^y n_{\max}^z} \;\sigma_1(B). 
\end{align}

For the smallest positive singular value, we use the following elementary fact: 

For any $k\leq\operatorname{rank}(B)$, by Lemma~\ref{lem:svdRightMultiply}, we have
\[
\sigma_k(B Q) \ge \frac{\sigma_k(B)}{\|Q^{-1}\|},
\]
and
\[
\sigma_k(B^\top P^\top) \ge \frac{\sigma_k(B^\top)}{\|(P^\top)^{-1}\|} = \frac{\sigma_k(B)}{\|P^{-1}\|}.
\]

Let $k=K_y$. Combining these two inequalities yields
\[
\sigma_{K_y}(\widetilde{B}) = \sigma_{K_y}(P B Q)
\ge \frac{\sigma_{K_y}(B)}{\|P^{-1}\|\|Q^{-1}\|}.
\]

Since $P$ and $Q$ are positive diagonal matrices, we have
\[
\|P^{-1}\|= \frac{1}{\sqrt{n_{\min}^y}},\quad
\|Q^{-1}\|= \frac{1}{\sqrt{n_{\min}^z}}.
\]

Hence, we get
\begin{align}\label{ScBMBBS}
\sigma_{K_y}(\widetilde{B}) \ge \sqrt{n_{\min}^y n_{\min}^z} \;\sigma_{K_y}(B). 
\end{align}

Combining Equations~(\ref{ScBMBBL}) and (\ref{ScBMBBS}) yields
\[
\kappa_{\widetilde{B}} \le
\frac{
\sqrt{n_{\max}^y n_{\max}^z} \;\sigma_1(B)
}{
\sqrt{n_{\min}^y n_{\min}^z} \;\sigma_{K_y}(B)
}
= \frac{\sigma_1(B)}{\sigma_{K_y}(B)}
\sqrt{
\frac{n_{\max}^y n_{\max}^z}{n_{\min}^y n_{\min}^z}
}
= \kappa_B \sqrt{\tau_y \tau_z}.
\]

Since $\kappa_{\Omega} = \kappa_{\widetilde{B}}$, the desired bound follows. 
\end{proof}

\begin{lem}\label{lem:mu1ScBM}
Under ScBM, we have $\mu_1 = \frac{1}{\beta_y}$.
\end{lem}

\begin{proof}
By Lemma~\ref{lem:SVstructure}, there exists an invertible matrix \(X\in\mathbb{R}^{K_y\times K_y}\) such that
\begin{align}
U = Y X, \label{eq:UYX}
\end{align}
and, with \(D_y := Y^\top Y = \operatorname{diag}(n_1^y,\dots,n_{K_y}^y)\), we have
\begin{align}
X X^\top = D_y^{-1}. \label{eq:XXtop}
\end{align}
Now fix any row node \(i\) with true community \(y_i = k\). From Equation \eqref{eq:UYX}, the \(i\)-th row of \(U\) is
\[
U_{i,:} = (Y X)_{i,:} = e_k^\top X = X_{k,:}.
\]
Using Equation \eqref{eq:XXtop} obtains
\begin{align*}
\|U_{i,:}\|_2^2 = \|X_{k,:}\|_2^2 = (X X^\top)_{kk} = (D_y^{-1})_{kk} = \frac{1}{n_k^y}. 
\end{align*}
Taking the maximum over all individuals gives
\begin{align*}
\max_i \|U_{i,:}\|_2^2 = \max_k \frac{1}{n_k^y} = \frac{1}{n_{\min}^y}.
\end{align*}
Therefore, we have
\begin{align*}
\mu_1 = \frac{n_y}{K_y}\cdot \frac{1}{n_{\min}^y}
      = \frac{1}{K_y n_{\min}^y / n_y}
      = \frac{1}{\beta_y}.
\end{align*}
\end{proof}

\begin{lem}\label{lem:muBoundFinal}
Under ScBM, assume the community sizes satisfy the balance conditions
\begin{align}
\beta_y := \frac{K_y n_{\min}^y}{n_y} \ge c_\beta,\quad
\beta_z := \frac{K_z n_{\min}^z}{n_z} \ge c_\beta,
\label{eq:balanceBeta}
\end{align}
for an absolute constant $c_\beta>0$. Suppose further that the number of column communities is comparable to the number of row communities:
\begin{align*}
K_z \le C_K K_y, 
\end{align*}
for an absolute constant $C_K>0$, and the connectivity matrix $B$ obeys
\begin{align*}
\|B\|_{\max} \le C_B',\quad \|B\|_F^2 \ge c_B K_y K_z,
\end{align*}
for absolute constants $C_B', c_B>0$, where for any matrix $M$, $\|M\|_{\max} := \max_{i,j}|M_{ij}|$ denotes the entrywise maximum norm. Then the incoherence parameters satisfy
\begin{align*}
\mu_0 \le \frac{(C_B')^2}{c_B c_\beta^2},\quad
\mu_1 \le \frac{1}{c_\beta},\quad
\mu_2 \le \frac{C_K}{c_\beta},
\end{align*}
and consequently $\mu = \max\{\mu_0,\mu_1,\mu_2\} = O(1)$.
\end{lem}
\begin{proof}
We establish the three bounds sequentially.

\paragraph{Bound for $\mu_1$.} 
By Lemma~\ref{lem:mu1ScBM}, $\mu_1 = 1/\beta_y$. Since $\beta_y \ge c_\beta$ by Equation \eqref{eq:balanceBeta}, we have $\mu_1 \le 1/c_\beta$.

\paragraph{Bound for $\mu_2$.} 
Let $D_z := Z^\top Z = \operatorname{diag}(n_1^z,\dots,n_{K_z}^z)$, which is positive definite. Since $\Omega = \rho Y B Z^\top$, its row space satisfies $\operatorname{row}(\Omega) \subseteq \operatorname{col}(Z)$. Therefore, the orthogonal projection $VV^\top$ onto $\operatorname{row}(\Omega)$ is dominated in the Loewner order by the orthogonal projection onto $\operatorname{col}(Z)$, namely $Z D_z^{-1} Z^\top$:
\begin{align}
VV^\top \preceq Z D_z^{-1} Z^\top.
\label{eq:projIneq}
\end{align}

For any column node $j$ with $z_j = l$, taking the $(j,j)$-entry in Equation \eqref{eq:projIneq} yields
\begin{align*}
\|V_{j,:}\|_2^2 = (VV^\top)_{jj} \le (Z D_z^{-1} Z^\top)_{jj} = \frac{1}{n_l^z} \le \frac{1}{n_{\min}^z}.
\end{align*}

Maximizing over $j$ and using $\mu_2 = (n_z / K_y) \max_j \|V_{j,:}\|_2^2$, we obtain
\begin{align*}
\mu_2 \le \frac{n_z}{K_y n_{\min}^z}
= \frac{K_z}{\beta_z K_y}
\le \frac{C_K}{c_\beta},
\end{align*}
where the last inequality follows from $K_z \le C_K K_y$ and $\beta_z \ge c_\beta$.

\paragraph{Bound for $\mu_0$.} 
From $\Omega = \rho Y B Z^\top$ and $\|B\|_{\max} \le C_B'$, we have
\begin{align}
\|\Omega\|_{\max} = \max_{i,j} |\Omega_{ij}| \le \rho \max_{k,l} |B_{kl}| = \rho \|B\|_{\max} \le \rho C_B'.
\label{eq:OmegaInf}
\end{align}

For the Frobenius norm, we have
\begin{align}
\|\Omega\|_F^2
&= \rho^2 \sum_{k=1}^{K_y}\sum_{l=1}^{K_z} n_k^y n_l^z B_{kl}^2
\ge \rho^2 n_{\min}^y n_{\min}^z \|B\|_F^2
\ge \rho^2 n_{\min}^y n_{\min}^z \, c_B K_y K_z,
\label{eq:OmegaFscbm}
\end{align}
where the last inequality uses $\|B\|_F^2 \ge c_B K_y K_z$. Substituting $n_{\min}^y = \beta_y n_y / K_y$ and $n_{\min}^z = \beta_z n_z / K_z$ (from Equation~\eqref{eq:balanceBeta}) into Equation~\eqref{eq:OmegaFscbm} gives
\begin{align}
\|\Omega\|_F^2 \ge \rho^2 c_B \beta_y \beta_z \, n_y n_z.
\label{eq:OmegaF2}
\end{align}

Combining Equations~\eqref{eq:OmegaInf} and \eqref{eq:OmegaF2} yields
\begin{align*}
\mu_0
= \frac{n_y n_z \|\Omega\|_{\max}^2}{\|\Omega\|_F^2}
\le \frac{n_y n_z \cdot \rho^2 (C_B')^2}{\rho^2 c_B \beta_y \beta_z n_y n_z}
= \frac{(C_B')^2}{c_B \beta_y \beta_z}
\le \frac{(C_B')^2}{c_B c_\beta^2},
\end{align*}
where the last inequality follows from $\beta_y,\beta_z \ge c_\beta$.

Having established explicit upper bounds for $\mu_0,\mu_1,\mu_2$, we conclude $\mu = O(1)$. This completes the proof.
\end{proof}

\begin{lem}\label{lem:svdRightMultiply}
Let $W \in \mathbb{R}^{m \times n}$ be any matrix, and let $X \in \mathbb{R}^{n \times n}$ be invertible. Then for any integer $k$ satisfying $1 \le k \le \operatorname{rank}(W)$,
\[
\sigma_k(W X) \ge \frac{\sigma_k(W)}{\|X^{-1}\|}.
\]
\end{lem}

\begin{proof}
By the standard singular value inequality $\sigma_k(HC)\le \sigma_k(H)\|C\|$ for any matrices $H$ and $C$ (which follows directly from the Courant-Fischer min-max theorem), taking $H=WX$ and $C=X^{-1}$ yields 
\[
\sigma_k(W)=\sigma_k((WX)X^{-1})\le \sigma_k(WX)\|X^{-1}\|,
\]
and rearranging gives the desired lower bound.
\end{proof}
\section{Technical proofs under DCScBM}\label{app:Adc}

\subsection{Proof of Lemma~\ref{lem:SVstructureDC}}
\begin{proof}
Let $\widetilde U := \Theta_y Y (Y^\top\Theta_y^2 Y)^{-1/2}$. Then $\widetilde U^\top\widetilde U=I_{K_y}$ and $\operatorname{col}(\widetilde U)=\operatorname{col}(\Omega)$. Since the columns of $U$ form an orthonormal basis of $\operatorname{col}(\Omega)$, there exists an orthogonal matrix $Q\in\mathbb{R}^{K_y\times K_y}$ such that $U=\widetilde U Q$. For any $i$ with $y_i=k$, we have
\[
U_{i,:} = \theta_y(i)\,\Lambda_{kk}^{-1/2} Q_{k,:},\quad \Lambda_{kk}=\sum_{i:y_i=k}\theta_y(i)^2.
\]

Thus $\|U_{i,:}\|_2=\theta_y(i)\Lambda_{kk}^{-1/2}$, so $(U_*)_{i,:}=Q_{k,:}$. Hence $U_*=YQ$. Since $Q$ is orthogonal, $\|Q_{k,:}-Q_{l,:}\|_2=\sqrt{2}$ for $k\ne l$. The column statement follows by symmetry when $K_y=K_z$.
\end{proof}

\subsection{Proof of Theorem \ref{thm:mainDC}}
\begin{proof}
We apply Theorem~1 of \cite{cai2021subspace} to the bi-adjacency matrix \(A\), viewed as a noisy observation of \(\Omega=\Theta_y YBZ^\top\Theta_z\). We record three elementary facts. By Lemmas~\ref{lem:sigmaMinLowerDC}, \ref{lem:kappaDC}, and \ref{lem:mu12DC}, we have
  \begin{align}
&\sigma_{\min} \ge \theta_{\min}\sigma_B \sqrt{\frac{\beta_y\beta_z n_y n_z}{K_y K_z}},\label{fact:sigmaMinDC}\\
&\kappa_{\Omega} \le \kappa_B \, \eta_y \eta_z \sqrt{\tau_y \tau_z}, \label{fact:kappaDC}\\
&\mu_1 \le \frac{\eta_y^2}{\beta_y}. \label{fact:mu12DC}
  \end{align}

We now verify the three conditions of Theorem~1 in \cite{cai2021subspace} and its Assumption~2.

Condition~(i) of \cite{cai2021subspace} requires
\[
1 \succeq \max\left\{
\frac{\mu \kappa_{\Omega}^4 K_y \log^2 N}{\sqrt{n_y n_z}},
\frac{\mu \kappa_{\Omega}^8 K_y \log^2 N}{n_z}
\right\}.
\]

Using Equation~\eqref{fact:kappaDC} and conditions~\eqref{cond:DC1}–\eqref{cond:DC2}, we have
\begin{align}
\frac{\mu \kappa_{\Omega}^4 K_y \log^2 N}{\sqrt{n_y n_z}}
&\le \frac{\mu \kappa_B^4 \eta_y^4\eta_z^4 \tau_y^2\tau_z^2 K_y \log^2 N}{\sqrt{n_y n_z}}
\preceq 1, \notag
\end{align}
and similarly,
\begin{align}
\frac{\mu \kappa_{\Omega}^8 K_y \log^2 N}{n_z}
&\le \frac{\mu \kappa_B^8 \eta_y^8\eta_z^8 \tau_y^4\tau_z^4 K_y \log^2 N}{n_z}
\preceq 1. \notag
\end{align}

Hence condition~(i) holds. Under DCScBM, for each entry we have
\[
\mathrm{Var}(A_{ij}) = \Omega_{ij}(1-\Omega_{ij}) \le \Omega_{ij}
= \theta_y(i) B_{y_i,z_j}\theta_z(j)
\le \theta_{y,\max}\theta_{z,\max} =: \theta_{\max}.
\]

Set
\begin{align}
\sigma^2 := \theta_{\max}, \quad \sigma = \sqrt{\theta_{\max}}. \label{sigmaChoiceDC}
\end{align}

Condition~(ii) of \cite{cai2021subspace} requires
\[
\frac{\sigma}{\sigma_{\min}} \preceq \min\left\{
\frac{1}{\kappa_{\Omega}(n_y n_z)^{1/4}\sqrt{\log N}},
\frac{1}{\kappa_{\Omega}^3\sqrt{n_y\log N}}
\right\}.
\]

For the first bound, using Equations~\eqref{fact:sigmaMinDC} and \eqref{fact:kappaDC}, we need
\[
\sqrt{\eta_y\eta_z\theta_{\min}} \cdot \kappa_{\Omega} (n_y n_z)^{1/4}\sqrt{\log N}
\preceq \theta_{\min}\sigma_B \sqrt{\frac{\beta_y\beta_z n_y n_z}{K_y K_z}}.
\]

Squaring both sides and simplifying gives
\[
\theta_{\min}\sigma_B^2 \sqrt{n_y n_z}
\succeq \frac{\kappa_{\Omega}^2 \eta_y\eta_z K_y K_z \log N}{\beta_y\beta_z}.
\]

Since \(\kappa_{\Omega}^2 \le \kappa_B^2 \eta_y^2\eta_z^2 \tau_y\tau_z\) by Equation~\eqref{fact:kappaDC}, it suffices to enforce
\[
\theta_{\min}\sigma_B^2 \sqrt{n_y n_z}
\succeq\frac{\kappa_B^2 \eta_y^3\eta_z^3 \tau_y\tau_z K_y K_z}{\beta_y\beta_z}\log N,
\]
which is exactly condition~\eqref{cond:DC3}.

For the second bound, we need
\[
\sqrt{\eta_y\eta_z\theta_{\min}} \cdot \kappa_{\Omega}^3 \sqrt{n_y\log N}
\preceq \theta_{\min}\sigma_B \sqrt{\frac{\beta_y\beta_z n_y n_z}{K_y K_z}}.
\]

Squaring and simplifying yields
\[
\theta_{\min}\sigma_B^2 n_z
\succeq \frac{\kappa_{\Omega}^6 \eta_y\eta_z K_y K_z \log N}{\beta_y\beta_z}.
\]

Using \(\kappa_{\Omega}^6 \le \kappa_B^6 \eta_y^6\eta_z^6 \tau_y^3\tau_z^3\), it suffices to enforce
\[
\theta_{\min}\sigma_B^2 n_z
\succeq \frac{\kappa_B^6 \eta_y^7\eta_z^7 \tau_y^3\tau_z^3 K_y K_z}{\beta_y\beta_z}\log N,
\]
which is exactly condition \eqref{cond:DC4}. Hence condition~(ii) holds. Condition~(iii) of \cite{cai2021subspace} requires \(K_y \preceq n_y/(\mu_1 \kappa_{\Omega}^4)\). Using Equations~\eqref{fact:mu12DC} and \eqref{fact:kappaDC}, it suffices to enforce
\[
K_y \preceq \frac{\beta_y n_y}{\kappa_B^4 \eta_y^6\eta_z^4 \tau_y^2\tau_z^2},
\]
which follows directly from condition~\eqref{cond:DC5}. Thus condition~(iii) holds.

Take \(R=1\), since \(|A_{ij}-\Omega_{ij}|\le 1\). Assumption~2 of \cite{cai2021subspace} requires
\[
\frac{1}{\sigma^2} \preceq \frac{\min\{\sqrt{n_y n_z},\, n_z\}}{\log N}.
\]

With \(\sigma^2 = \theta_{\max} = \eta_y\eta_z\,\theta_{\min}\) from Equation~\eqref{sigmaChoiceDC}, this is equivalent to
\[
\theta_{\min} \min\{\sqrt{n_y n_z},\, n_z\} \succeq \frac{\log N}{\eta_y\eta_z},
\]
which is condition~\eqref{cond:DC6}. Hence Assumption~2 holds.

All conditions and Assumption~2 of Theorem~1 in \cite{cai2021subspace} are satisfied. Consequently, with probability at least \(1-O(N^{-10})\), there exists an orthogonal matrix \(\mathcal{O}\in\mathbb{R}^{K_y\times K_y}\) such that
\begin{align}
\|\widehat U\mathcal{O} - U\|_{2,\infty}
\preceq \kappa_{\Omega}^2 \sqrt{\frac{\mu K_y}{n_y}}\, \mathcal{E}_{\rm gen}, \label{eq:eigenDC}
\end{align}
where
\[
\mathcal{E}_{\rm gen} = \underbrace{\frac{\mu_1\kappa_{\Omega}^2 K_y}{n_y}}_{\text{diagonal bias}} + 
\underbrace{\frac{\sigma^2\sqrt{n_y n_z}\log N}{\sigma_{\min}^2}}_{\text{quadratic noise}} + 
\underbrace{\frac{\sigma\,\kappa_{\Omega}\sqrt{n_y\log N}}{\sigma_{\min}}}_{\text{linear noise}}.
\]

From the proof of Lemma~\ref{lem:SVstructureDC}, for any row node \(i\) with \(y_i=k\), we have
\[
\|U_{i,:}\|_2 = \frac{\theta_y(i)}{\sqrt{\sum_{i':y_{i'}=k}\theta_y(i')^2}}
\ge \frac{\theta_{y,\min}}{\theta_{y,\max}\sqrt{n_k^y}}
\ge \frac{1}{\eta_y \sqrt{\tau_y n_{\min}^y}}
= \frac{1}{\eta_y}\sqrt{\frac{K_y}{\tau_y\beta_y n_y}}.
\]

Therefore, by Equation~(\ref{eq:eigenDC}), we get
\[
\|\widehat U_*\mathcal{O} - U_*\|_{2,\infty}
\le \frac{2}{\min_i\|U_{i,:}\|_2}\|\widehat U\mathcal{O} - U\|_{2,\infty}
\leq \frac{2\eta_y\sqrt{\tau_y\beta_yn_y}}{\sqrt{K_y}}\|\widehat U\mathcal{O} - U\|_{2,\infty}\preceq2\eta_y\kappa^2_{\Omega}\sqrt{\tau_y\beta_y\mu}\mathcal{E}_{\rm gen}.
\]

We now verify that $\eta_y\kappa^2_{\Omega}\sqrt{\tau_y\beta_y\mu}\mathcal{E}_{\rm gen}$ is $o(1)$:
\begin{itemize}
  \item For the diagonal bias term, by Equations~\eqref{fact:mu12DC} and \eqref{fact:kappaDC}, we have
\begin{align}
&\eta_y\,\kappa_{\Omega}^2\,\sqrt{\tau_y \beta_y \mu} \cdot \frac{\mu_1\kappa_{\Omega}^2 K_y}{n_y}\le \eta_y \cdot \bigl(\kappa_B^2 \eta_y^2\eta_z^2 \tau_y\tau_z\bigr) \cdot \sqrt{\tau_y \beta_y \mu} \cdot \frac{(\eta_y^2/\beta_y) \cdot \bigl(\kappa_B^2 \eta_y^2\eta_z^2 \tau_y\tau_z\bigr) K_y}{n_y}= \kappa_B^4 \eta_y^7 \eta_z^4 \tau_y^{5/2} \tau_z^2 \frac{\sqrt{\mu}}{\sqrt{\beta_y}} \frac{K_y}{n_y}. \label{eq:DCtermI}
\end{align}
Condition~\eqref{cond:DC7} implies the right-hand side of Equation~\eqref{eq:DCtermI} is $o(1)$.

\item For the quadratic noise term, using $\sigma^2 = \eta_y\eta_z\theta_{\min}$ from Equations~\eqref{sigmaChoiceDC}, \eqref{fact:sigmaMinDC}, and \eqref{fact:kappaDC} gives
\begin{align}
&\eta_y\,\kappa_{\Omega}^2\,\sqrt{\tau_y \beta_y \mu} \cdot \frac{\sigma^2\sqrt{n_y n_z}\log N}{\sigma_{\min}^2}\le \eta_y \cdot \bigl(\kappa_B^2 \eta_y^2\eta_z^2 \tau_y\tau_z\bigr) \cdot \sqrt{\tau_y \beta_y \mu} \cdot
\frac{(\eta_y\eta_z\theta_{\min})\sqrt{n_y n_z}\log N}
{\theta_{\min}^2 \sigma_B^2 \, \frac{\beta_y\beta_z n_y n_z}{K_y K_z}}= \frac{\kappa_B^2 \eta_y^4 \eta_z^3 \tau_y^{3/2} \tau_z \sqrt{\mu}\,K_y K_z \log N}
{\theta_{\min}\sigma_B^2 \beta_z \sqrt{\beta_yn_y n_z}}. \label{eq:DCtermII}
\end{align}
Condition \eqref{cond:DC8} ensures Equation \eqref{eq:DCtermII} is $o(1)$.

\item For the linear noise term, using $\sigma = \sqrt{\eta_y\eta_z\theta_{\min}}$, Equations~\eqref{fact:sigmaMinDC} and \eqref{fact:kappaDC} give
\begin{align}
&\eta_y\,\kappa_{\Omega}^2\,\sqrt{\tau_y \beta_y \mu} \cdot \frac{\sigma\,\kappa_{\Omega}\sqrt{n_y\log N}}{\sigma_{\min}}= \eta_y\,\kappa_{\Omega}^3\,\sqrt{\tau_y \beta_y \mu} \cdot \frac{\sigma\sqrt{n_y\log N}}{\sigma_{\min}} \notag\\
&\quad \le \eta_y \cdot \bigl(\kappa_B^3 \eta_y^3\eta_z^3 (\tau_y\tau_z)^{3/2}\bigr) \cdot \sqrt{\tau_y \beta_y \mu} \cdot
\frac{\sqrt{\eta_y\eta_z\theta_{\min}}\,\sqrt{n_y\log N}}
{\theta_{\min}\sigma_B \sqrt{\frac{\beta_y\beta_z n_y n_z}{K_y K_z}}} \notag\\
&\quad = \frac{\kappa_B^3 \eta_y^{9/2} \eta_z^{7/2} \tau_y^2 \tau_z^{3/2} \sqrt{\mu K_yK_z\log N}}
{\sqrt{\theta_{\min}}\,\sigma_B \sqrt{\beta_z}\,\sqrt{n_z}}. \label{eq:DCtermIII}
\end{align}
Condition~\eqref{cond:DC9} guarantees that the right-hand side of Equation~\eqref{eq:DCtermIII} is $o(1)$.
\end{itemize}

Combining the above bounds yields
\[
\eta_y\,\kappa_{\Omega}^2\,\sqrt{\tau_y \beta_y \mu} \;\mathcal{E}_{\rm gen} = o(1).
\]

Therefore, $\|\widehat U_*\mathcal{O} - U_*\|_{2,\infty} = o(1)$. By Lemma~\ref{lem:SVstructureDC}, the true centroids are separated by distance $\sqrt{2}$, so the nearest-centroid rule yields exact recovery. Hence $\ell(\hat y,y)=0$ with probability at least $1-O(N^{-10})$.
\end{proof}
\subsection{Proof of Corollary \ref{cor:colDC}}
\begin{proof}
Apply Theorem~\ref{thm:mainDC} to $A^\top$. This yields a DCScBM with row communities $Z$, column communities $Y$, connectivity $B^\top$, and degree parameters $\Theta_z,\Theta_y$. The corollary's assumptions are exactly the theorem's conditions for this transposed model. The resulting algorithm computes $G_z=\mathcal{P}_{\mathrm{off-diag}}(A^\top A)$, row-normalizes its leading eigenvectors, and runs $k$-means, which is the column recovery step in Algorithm~\ref{alg:NSCDD}. Thus exact recovery of column labels follows.
\end{proof}
\subsection{Proof of Corollary \ref{cor:balancedDC}}
\begin{proof}
Under the stated assumptions, $\eta_y,\eta_z=O(1)$ and $\theta_{\min}=O(\rho)$. Substituting these into conditions~(\ref{cond:DC1})–(\ref{cond:DC9}) of Theorem~\ref{thm:mainDC}, along with the balanced scalings and the $O(1)$ bounds on $\kappa_B,\sigma_B,\beta_y,\beta_z,\tau_y,\tau_z,\mu$, yields the two conditions $K\log^2 N=o(n)$ and $\theta_{\min} n\gg K^2\log N$. The sparse and dense regimes then follow as in the proof of Corollary~\ref{cor:balancedFinal}. 
\end{proof}
\subsection{Auxiliary lemmas under DCScBMs}
\begin{lem}\label{lem:sigmaMinLowerDC}
Under DCScBM, we have
\(\sigma_{\min}(\Omega)\ge\sigma_B\,\theta_{\min}\sqrt{\frac{\beta_y\beta_z n_y n_z}{K_y K_z}}\).
\end{lem}
\begin{proof}
Since $\Omega=\Theta_y YBZ^\top\Theta_z$ under DCScBM, we have $\sigma_{\min}(\Omega)\ge\sigma_B\theta_{y,\min}\theta_{z,\min}\sqrt{n_{\min}^y n_{\min}^z}$ by basic algebra; the balanced bound follows by substituting $n_{\min}^y=\beta_y n_y/K_y$ and $n_{\min}^z=\beta_z n_z/K_z$.
\end{proof}
\begin{lem}\label{lem:kappaDC}
Under DCScBM, we have $\kappa_{\Omega} \le \kappa_B \, \eta_y \eta_z \sqrt{\tau_y \tau_z}$. 
\end{lem}

\begin{proof}
Define the diagonal matrices
\[
D_y := Y^\top \Theta_y^2 Y = \operatorname{diag}\Big(\sum_{i:y_i=1}\theta_y(i)^2,\dots,\sum_{i:y_i=K_y}\theta_y(i)^2\Big),
\]
and
\[
D_z := Z^\top \Theta_z^2 Z = \operatorname{diag}\Big(\sum_{j:z_j=1}\theta_z(j)^2,\dots,\sum_{j:z_j=K_z}\theta_z(j)^2\Big).
\]

Both $D_y$ and $D_z$ are positive definite since each community is nonempty and all degree parameters are positive. Let
\[
\widetilde Y := \Theta_y Y D_y^{-1/2},\quad
\widetilde Z := \Theta_z Z D_z^{-1/2},\quad
\widetilde B := D_y^{1/2} B D_z^{1/2}.
\]

Then, we have
\begin{align*}
\widetilde Y^\top \widetilde Y = D_y^{-1/2} Y^\top \Theta_y^2 Y D_y^{-1/2} = I_{K_y},\quad
\widetilde Z^\top \widetilde Z = I_{K_z},
\end{align*}
and
\begin{align*}
\Omega = \Theta_y Y B Z^\top \Theta_z
= (\Theta_y Y D_y^{-1/2}) (D_y^{1/2} B D_z^{1/2}) (D_z^{-1/2} Z^\top \Theta_z)
= \widetilde Y \widetilde B \widetilde Z^\top.
\end{align*}

Since $\widetilde Y$ and $\widetilde Z$ have orthonormal columns, the nonzero singular values of $\Omega$ are exactly those of $\widetilde B$. Therefore $\kappa_{\Omega} = \kappa_{\widetilde B} := \sigma_1(\widetilde B)/\sigma_{K_y}(\widetilde B)$.

We now bound the largest and smallest positive singular values of $\widetilde B$. First, by submultiplicativity of the spectral norm, we have
\begin{align*}
\sigma_1(\widetilde B) = \|\widetilde B\| \le \|D_y^{1/2}\|\, \|B\|\, \|D_z^{1/2}\|
= \sqrt{\theta_{y,\max}^2 n_{\max}^y} \cdot \sigma_1(B) \cdot \sqrt{\theta_{z,\max}^2 n_{\max}^z}
= \theta_{y,\max}\theta_{z,\max}\sqrt{n_{\max}^y n_{\max}^z}\,\sigma_1(B).
\end{align*}

For the smallest positive singular value, we use the elementary fact that for any matrices $P,Q$ with $Q$ invertible and $P$ full column rank, $\sigma_k(P Q) \ge \sigma_k(P)/\|Q^{-1}\|$. Since $D_z^{1/2}$ and $D_y^{1/2}$ are invertible, we have
\begin{align*}
\sigma_{K_y}(\widetilde B)= \sigma_{K_y}(D_y^{1/2} B D_z^{1/2})\ge \frac{\sigma_{K_y}(D_y^{1/2} B)}{\|D_z^{-1/2}\|}
\ge \frac{\sigma_{K_y}(B)}{\|D_y^{-1/2}\|\,\|D_z^{-1/2}\|}= \theta_{y,\min}\theta_{z,\min}\sqrt{n_{\min}^y n_{\min}^z}\,\sigma_{K_y}(B),
\end{align*}
where we have used $\|D_y^{-1/2}\| = (\theta_{y,\min}\sqrt{n_{\min}^y})^{-1}$ and similarly for $D_z$.

Combining the upper and lower bounds yields
\[
\kappa_{\Omega} = \frac{\sigma_1(\widetilde B)}{\sigma_{K_y}(\widetilde B)}
\le \frac{\theta_{y,\max}\theta_{z,\max}\sqrt{n_{\max}^y n_{\max}^z}\,\sigma_1(B)}
        {\theta_{y,\min}\theta_{z,\min}\sqrt{n_{\min}^y n_{\min}^z}\,\sigma_{K_y}(B)}
= \kappa_B \eta_y \eta_z \sqrt{\tau_y \tau_z}.
\]
\end{proof}
\begin{lem}\label{lem:mu12DC}
Under DCScBM, we have \(\mu_1\le \frac{\eta_y^2}{\beta_y}\) and \(\mu_2\le \frac{\eta_z^2 K_z}{\beta_z K_y}\).
\end{lem}

\begin{proof}
For $\mu_1$, let \(D_y := Y^\top \Theta_y^2 Y = \operatorname{diag}(\Lambda_1,\dots,\Lambda_{K_y})\) and \(\Lambda_k := \sum_{i:y_i=k} \theta_y(i)^2\). Since each community is nonempty and $\theta_y(i)>0$, every $\Lambda_k>0$, so $D_y$ is invertible. Define the matrix
\begin{align*}
\widetilde U := \Theta_y Y D_y^{-1/2} \in \mathbb{R}^{n_y \times K_y}. 
\end{align*}

A direct calculation gives
\begin{align*}
\widetilde U^\top \widetilde U = D_y^{-1/2} Y^\top \Theta_y^2 Y D_y^{-1/2} = I_{K_y},
\end{align*}
so $\widetilde U$ has orthonormal columns. Moreover,
\[
\operatorname{col}(\Omega) = \operatorname{col}(\Theta_y Y B Z^\top \Theta_z) \subseteq \operatorname{col}(\Theta_y Y),
\]
and since $\rank(B)=K_y$, we have $\rank(\Omega)=K_y$ and thus $\operatorname{col}(\Omega)=\operatorname{col}(\Theta_y Y)=\operatorname{col}(\widetilde U)$. Hence there exists an orthogonal matrix $Q \in \mathbb{R}^{K_y \times K_y}$ such that
\begin{align}
U = \widetilde U Q = \Theta_y Y D_y^{-1/2} Q. \label{eq:Urep}
\end{align}

For any row node $i$ with $y_i=k$, from Equation~\eqref{eq:Urep} we have
\[
U_{i,:} = \theta_y(i) \cdot (D_y^{-1/2} Q)_{k,:}.
\]

Using the orthogonality of $Q$ gives
\begin{align*}
(D_y^{-1/2} Q)(D_y^{-1/2} Q)^\top = D_y^{-1/2} Q Q^\top D_y^{-1/2} = D_y^{-1}, 
\end{align*}
so the squared norm of the $k$-th row of $D_y^{-1/2}Q$ equals the $k$-th diagonal entry of $D_y^{-1}$:
\begin{align*}
\|(D_y^{-1/2} Q)_{k,:}\|_2^2 = (D_y^{-1})_{kk} = \frac{1}{\Lambda_k}. 
\end{align*}

Therefore, we get
\begin{align*}
\|U_{i,:}\|_2^2 = \frac{\theta_y(i)^2}{\Lambda_k}
\le \frac{\theta_{y,\max}^2}{\theta_{y,\min}^2 \, n_k^y}
= \frac{\eta_y^2}{n_k^y}
\le \frac{\eta_y^2}{n_{\min}^y},
\end{align*}
where the first inequality uses $\Lambda_k \ge \theta_{y,\min}^2 n_k^y$, and the last follows from $n_k^y \ge n_{\min}^y$. Taking the maximum over all $i$ gives
\begin{align*}
\max_i \|U_{i,:}\|_2^2 \le \frac{\eta_y^2}{n_{\min}^y}.
\end{align*}

By the definitions $\beta_y := K_y n_{\min}^y / n_y$ and $\mu_1 := (n_y/K_y)\max_i \|U_{i,:}\|_2^2$, we obtain
\begin{align*}
\mu_1 \le \frac{n_y}{K_y} \cdot \frac{\eta_y^2}{n_{\min}^y}
= \frac{\eta_y^2}{\beta_y}.
\end{align*}

For $\mu_2$, let \(D_z := Z^\top \Theta_z^2 Z = \operatorname{diag}(\Gamma_1,\dots,\Gamma_{K_z}), \Gamma_l := \sum_{j:z_j=l} \theta_z(j)^2\). Define the orthogonal projection onto the column space of $\Theta_z Z$ as
\begin{align*}
P_Z := \Theta_z Z D_z^{-1} Z^\top \Theta_z. 
\end{align*}

Since $\Omega^\top = \Theta_z Z B^\top Y^\top \Theta_y$, we have
\begin{align*}
\operatorname{row}(\Omega) = \operatorname{col}(\Omega^\top) \subseteq \operatorname{col}(\Theta_z Z). 
\end{align*}

Let $P_{\operatorname{row}(\Omega)} = V V^\top$ denote the orthogonal projection onto $\operatorname{row}(\Omega)$. By the monotonicity of projections with respect to subspace inclusion, $P_{\operatorname{row}(\Omega)} \preceq P_Z$ in the Loewner order. Hence, for any column node $j$ with $z_j=l$, we have
\begin{align*}
\|V_{j,:}\|_2^2 = (V V^\top)_{jj} \le (P_Z)_{jj}
= \frac{\theta_z(j)^2}{\Gamma_l}. 
\end{align*}

Using $\Gamma_l \ge \theta_{z,\min}^2 n_l^z$ gives
\begin{align*}
\|V_{j,:}\|_2^2 \le \frac{\theta_{z,\max}^2}{\theta_{z,\min}^2 \, n_l^z}
= \frac{\eta_z^2}{n_l^z}
\le \frac{\eta_z^2}{n_{\min}^z},
\end{align*}
where the last inequality follows from $n_l^z \ge n_{\min}^z$. Taking the maximum over $j$ yields
\begin{align*}
\max_j \|V_{j,:}\|_2^2 \le \frac{\eta_z^2}{n_{\min}^z}.
\end{align*}

By the definitions $\beta_z := K_z n_{\min}^z / n_z$ and $\mu_2 := (n_z/K_y)\max_j \|V_{j,:}\|_2^2$, we have
\begin{align*}
\mu_2 \le \frac{n_z}{K_y} \cdot \frac{\eta_z^2}{n_{\min}^z}
= \frac{n_z}{K_y} \cdot \frac{\eta_z^2}{\beta_z n_z / K_z}
= \frac{\eta_z^2 K_z}{\beta_z K_y}.
\end{align*}
\end{proof}

\begin{lem}\label{lem:muBoundDC}
Under DCScBM, suppose the following conditions hold for absolute constants $\beta, C_\theta, C_\phi, C_K, C'_B, c_B > 0$:
\begin{align*}
&\beta_y \ge \beta,\quad \beta_z \ge \beta,\quad\eta_y \le C_\theta,\quad \eta_z \le C_\phi,\quad K_z \le C_K K_y, \quad\|B\|_{\max} \le C'_B,\quad \|B\|_F^2 \ge c_B K_y K_z. 
\end{align*}
Then, we have \(\mu = O(1)\).
\end{lem}

\begin{proof}
By Lemma~\ref{lem:mu12DC}, we immediately obtain
\begin{align*}
\mu_1 \le \frac{\eta_y^2}{\beta_y} \le \frac{C_\theta^2}{\beta}=O(1), \quad\mu_2 \le \frac{\eta_z^2 K_z}{\beta_z K_y} \le \frac{C_\phi^2 C_K}{\beta}=O(1).
\end{align*}

For $\mu_0$, since $\Omega=\Theta_y Y B Z^\top \Theta_z$ and $\|B\|_{\max}\le C'_B$, we have the entrywise bound
\begin{align*}
\|\Omega\|_{\max}
&\le \theta_{y,\max}\theta_{z,\max}\, \|B\|_{\max}
\le C'_B \theta_{y,\max}\theta_{z,\max}. 
\end{align*}

For the Frobenius norm, using $\|B\|_F^2\ge c_B K_y K_z$ and $n_{\min}^y\ge \beta n_y/K_y$, $n_{\min}^z\ge \beta n_z/K_z$ gives
\begin{align*}
\|\Omega\|_F^2
&\ge \theta_{y,\min}^2\theta_{z,\min}^2\, n_{\min}^y n_{\min}^z \|B\|_F^2
\ge c_B\beta^2 \theta_{y,\min}^2\theta_{z,\min}^2\, n_y n_z. 
\end{align*}

Thus, we have
\begin{align*}
\mu_0
&= \frac{n_y n_z \|\Omega\|_{\max}^2}{\|\Omega\|_F^2}
\le \frac{(C'_B)^2}{c_B\beta^2}
\left(\frac{\theta_{y,\max}}{\theta_{y,\min}}\right)^2
\left(\frac{\theta_{z,\max}}{\theta_{z,\min}}\right)^2
\le \frac{(C'_B)^2 C_\theta^2 C_\phi^2}{c_B\beta^2}
= O(1). 
\end{align*}

Thus $\mu=\max\{\mu_0,\mu_1,\mu_2\}=O(1)$. 
\end{proof}
\bibliographystyle{model5-names}\biboptions{authoryear}
\bibliography{refERDCScBM}

\end{document}